\documentclass[10pt,journal,compsoc]{IEEEtran}
\IEEEoverridecommandlockouts

\ifCLASSOPTIONcompsoc
  \usepackage[nocompress]{cite}
\else
  \usepackage{cite}
\fi
\usepackage{amsmath,amssymb,amsfonts}
\usepackage{algorithmic}
\usepackage{graphicx}
\usepackage{textcomp}
\usepackage[absolute]{textpos}

\usepackage{enumitem}
\usepackage{color}
\usepackage{multirow}
\usepackage{lipsum}
\usepackage{listings}
\usepackage{graphicx}
\usepackage{pdfpages}
\usepackage{verbatim}
\usepackage{pifont}
\usepackage{etoolbox}
\usepackage{textcomp}
\usepackage{hyperref}
\usepackage{url}
\usepackage{caption}
\usepackage{booktabs} 
\usepackage{siunitx} 
\usepackage{etoolbox}

\usepackage[normalem]{ulem}
\useunder{\uline}{\ul}{}

\usepackage[switch]{lineno} 

\usepackage[linesnumbered,titlenumbered,ruled,vlined,resetcount,algosection]{algorithm2e}
\usepackage{xcolor}
\usepackage{color}
\def\BibTeX{{\rm B\kern-.05em{\sc i\kern-.025em b}\kern-.08em
    T\kern-.1667em\lower.7ex\hbox{E}\kern-.125emX}}

\usepackage{amsthm}
\newtheorem{theorem}{Theorem}
\newtheorem{corollary}{Corollary}

\usepackage{hhline}
\usepackage{xcolor}
\usepackage{colortbl}
\usepackage{threeparttable}
\usepackage{subfig}
\usepackage{comment}
\usepackage{enumitem}
\usepackage{tikz}
\newcommand*\circled[1]{\tikz[baseline=(char.base)]{
            \node[shape=circle,fill,inner sep=0.8pt] (char) {\footnotesize\textcolor{white}{#1}};}}

\def\BibTeX{{\rm B\kern-.05em{\sc i\kern-.025em b}\kern-.08em
    T\kern-.1667em\lower.7ex\hbox{E}\kern-.125emX}}

\newcommand{\pname}{\mbox{ZCCL}}

\newcommand{\fzlight}{\mbox{\textit{fZ}-light}}

\begin{document}

\Urlmuskip=0mu plus 1mu

\title{ZCCL: Significantly Improving Collective Communication With Error-Bounded Lossy Compression}

\author{Jiajun Huang, Sheng Di,~\IEEEmembership{Senior Member,~IEEE}, Xiaodong Yu, Yujia Zhai, Zhaorui Zhang, Jinyang Liu, Xiaoyi Lu, Ken Raffenetti, Hui Zhou, Kai Zhao, Khalid Alharthi, Zizhong Chen,~\IEEEmembership{Senior Member,~IEEE}, Franck Cappello,~\IEEEmembership{Fellow,~IEEE}, Yanfei Guo, Rajeev Thakur,~\IEEEmembership{Fellow,~IEEE}

\IEEEcompsocitemizethanks{
\IEEEcompsocthanksitem Jiajun Huang, Yujia Zhai, and Zizhong Chen are affiliated with the University of California, Riverside, CA 92521. Sheng Di, Ken Raffenetti, Hui Zhou, Franck Cappello, Yanfei Guo, and Rajeev Thakur are affiliated with Argonne National Laboratory, Lemont, IL 60439. Xiaodong Yu is affiliated with Stevens Institute of Technology, Hoboken, NJ 07030. Zhaorui Zhang is affiliated with The Hong Kong Polytechnic University, Kowloon, Hong Kong. Jinyang Liu is affiliated with University of Houston,
Houston, TX 77204. Xiaoyo Lu is is affiliated with University of California, Merced, CA 95343. Kai Zhao is affiliated with Florida State University, Tallahassee, FL 32306. Khalid Alharthi is affiliated with Department Of Computer Science, College of Computing And Information Technology, University Of Bisha, Bisha 61922, P.O. Box 551, Saudi Arabia. 
}

}



\IEEEtitleabstractindextext{%

\begin{abstract}

With the ever-increasing computing power of supercomputers and the growing scale of scientific applications, the efficiency of MPI collective communication turns out to be a critical bottleneck in large-scale distributed and parallel processing. The large message size in MPI collectives is particularly concerning because it can significantly degrade overall parallel performance. To address this issue, prior research simply applies off-the-shelf fixed-rate lossy compressors in the MPI collectives, leading to suboptimal performance, limited generalizability, and unbounded errors. In this paper, we propose a novel solution, called {\pname}, which leverages error-bounded lossy compression to significantly reduce the message size, resulting in a substantial reduction in communication costs. The key contributions are three-fold. (1) We develop two general, optimized lossy-compression-based frameworks for both types of MPI collectives (collective data movement as well as collective computation), based on their particular characteristics. Our framework not only reduces communication costs but also preserves data accuracy. (2) We customize {\fzlight}, an ultra-fast error-bounded lossy compressor, to meet the specific needs of collective communication. (3) We integrate {\pname} into multiple collectives, such as Allgather, Allreduce, Scatter, and Broadcast, and perform a comprehensive evaluation based on real-world scientific application datasets. Experiments show that our solution outperforms the original MPI collectives as well as multiple baselines by 1.9--8.9$\times$.

\end{abstract}

}
\maketitle

\maketitle


\section{Introduction}
MPI collectives provide high-performance collective communication in distributed systems, making a significant impact on various research fields such as scientific applications, distributed machine learning, and others~\cite{awan2017s, wang2006pelegant, abadi2016tensorflow, ayala2019impacts, jain2019scaling, abdelmoniem2021efficient}. With the advent of exascale computing and deep learning applications, the demand for large-message MPI collectives has increased. For example, in image classification tasks, VGG19~\cite{simonyan2015very} and ResNet-50~\cite{he2016deep} have 143 million and 25 million parameters, respectively, with communication overheads of 83\% and 72\%~\cite{abdelmoniem2021efficient}. Therefore, optimizing MPI collectives for large messages has become essential~\cite{chunduri2018characterization, Bayatpour2018SALaR, patarasuk2009bandwidth}.

MPI collectives consist of both internode communication and intranode communication, and the former is often the major concern. The overall collective performance is usually limited by the efficiency of internode communication because of limited network bandwidth. Therefore, optimizing internode collective communication is critical to improving the overall performance of MPI collectives. This topic has been a focus of research for decades, with state-of-the-art algorithms achieving notable improvements. However, with the increasing demand for large-message MPI collectives, further optimization remains necessary~\cite{Alm05BlueGene, thakur2005optimization, patarasuk2009bandwidth}. Lossy compression \cite{Di2016SZ,Tao2017SZ,Zhao2020SZauto,Lindstrom2014ZFP} (rather than lossless compression \cite{Deutsch1996gzip, Gaillyzlib,Collet2015zstd}) is a promising solution to mitigate this MPI collective performance issue because of its ability to significantly reduce the message size.

Although lossy compression has been widely used to resolve many other scalability issues in high-performance computing, such as reducing memory footprint \cite{quant-compression}, reducing storage space \cite{nbody-compression,mdz}, and avoiding duplicated computation \cite{pastri}, only a few studies have explored its use in this direction, and all expose certain limitations. To elaborate, Zhou et al. \cite{Zhou2021GPUCOMPRESSION} proposed GPU-compression enhanced point-to-point communication by integrating MPC \cite{yang2015mpc} and 1D fixed-rate ZFP \cite{Lindstrom2014ZFP} into MVAPICH2 \cite{SHM-MVAPICH2}. Their approach, referred to as \textit{CPRP2P}, simply involved compressing the messages before transmission and decompressing them after reception, leading to significant performance overhead due to the non-negligible time required for compression and decompression. Meanwhile, Zhou et al. \cite{Zhou2022GPUCOMPRESSIONALLTOALL, Zhou2022HiPC} proposed several additional approaches to improve multiple MPI collectives using 1D fixed-rate ZFP \cite{Lindstrom2014ZFP} on GPUs. Their methods, however, focus on fixed-rate compression\footnote{Fixed-rate compression means that the lossy compression would be performed based on a user-specified fixed compression ratio.}, introducing two major limitations: (1) compression errors cannot be bounded, leading to an uncontrolled accuracy, and (2) the compression quality is considerably lower compared to the fixed-accuracy mode\footnote{Fixed-accuracy, also known as error-bounded lossy compression, compresses data based on a user-specified error bound.} in ZFP, as demonstrated by prior research~\cite{fraz, huang2023ccoll}.

The aforementioned limitations of compression-enabled MPI collective algorithms motivate us to develop a new efficient MPI collective framework which leverages lossy compression technique to significantly improve the MPI collective performance. However, this brings in three direct technical challenges. \circled{A} Devising a general framework that can effectively hide the communication cost and choose an appropriate timing to call lossy compression is non-trivial.
\circled{B} The data loss nature of the lossy compression brings up a critical concern on the accuracy of collective operations. \circled{C} Existing lossy compressors are not designed for the collective context, leading to suboptimal collective performance because of unnecessary overheads when they are directly applied in MPI collectives \cite{Di2016SZ, Tao2017SZ, Liang2018SZ, Yu2022SZx, Lindstrom2014ZFP}.

Our developed framework is named as compression-facilitated MPI collective framework ({\pname}), which can address the aforementioned limitations and challenges. To the best of our knowledge, this is the \textit{first-ever} framework that provides a general high-performance solution for compression-integrated MPI collectives. Moreover, this is the first accuracy-aware design, which ensures the accelerated collective performance with error-bounded lossy compression does not compromise data quality. To be more specific, our contributions include:
\begin{itemize}
    \item To address challenge \circled{A}, we introduce two efficient frameworks, which can significantly accelerate both types of MPI collectives. Specifically, the first framework notably diminishes the compression overhead in the collective data movement operations (e.g., Scatter, Bcast and All-gather), thereby achieving substantial performance improvement. The second one hides communication inside of compression in collective computation (e.g., Reduce-scatter), which in turn enhances the performance of collective computation. Moreover, these two frameworks can be combined together to speed up more advanced MPI collective operations such as All-reduce. 
    \item To address challenge \circled{B}, we devise several strategies to effectively control error propagation within the {\pname} framework. Specifically, we employ error-bounded lossy compression, regulate the number of compression operations, and develop an efficient method to resolve imbalances in collective communication caused by error-bounded lossy compression. Through in-depth mathematical analysis, we prove that our {\pname} framework achieves well-bounded data accuracy.
    \item To address challenge \circled{C}, we select the most suitable error-bounded lossy compressor for collective communication, considering factors such as compression throughput, compression ratio, and compression quality. We thoroughly analyze the newly developed {\fzlight} compressor~\cite{SZp} and compare it with SZx~\cite{Yu2022SZx}, which is used in the state-of-the-art C-Coll framework~\cite{huang2023ccoll}. After determining that {\fzlight} is the most suitable option, we customize it to meet the specific needs of collective communication. Specifically, we redesign the compression workflow of {\fzlight} and implement a pipelined version to overlap compression and communication. This optimization reduces the communication cost in our {\pname} framework by up to 3.3$\times$.
    \item To demonstrate the generality of our design, we integrate {\pname} into multiple collectives, including Allgather, Allreduce, Scatter, and Broadcast, and evaluate performance using real-world scientific application datasets. Experiments on 128 Intel Broadwell compute nodes from a supercomputer show that ZCCL-accelerated collectives achieve significant speedups over the CPRP2P and C-Coll baselines, outperforming MPI\_Allreduce, MPI\_Scatter, and MPI\_Bcast by up to 3.6$\times$, 5.4$\times$, and 8.9$\times$, respectively. Additionally, we validate the practical effectiveness of ZCCL-accelerated Z-Allreduce using a real-world use case (image stacking analysis), which shows up to 3.0$\times$ performance gain over MPI\_Allreduce while maintaining high data integrity and accuracy during collective operations.
    
\end{itemize}

The rest of the paper is organized as follows: we introduce background and related work in Section \ref{sec:background} and detail our design and optimization in Section \ref{sec-design_and_optimizations}. Evaluation results are presented in Section \ref{exp-setup-sec} followed by conclusion and future work in Section \ref{sec:conclusion}.

\section{Background and Related Work}
\label{sec:background}
In this section, we discuss the background and related work. We first introduce MPI collective communication, followed by a discussion on high-speed lossy compressors and their integration with MPI implementations. The focus of our study is on lossy compression. This emphasis is due to the significantly lower compression ratios observed with lossless methods when applied to scientific datasets~\cite{Di2016SZ,Tao2017SZ,Huang-Exploring-Wavelet-Transform}.

\subsection{MPI Collective Communication}
\label{sec:MPI-Collectives}
There are many types of MPI collective operations, which can be divided into two sub-categories --- collective data movement and collective computation according to their communication patterns.
\subsubsection{Collective data movement}
Collective data movement includes gather, allgather, scatter, all-to-all, and so on. The gather operation collects the data from different processes and stores the collected data into the root process. In comparison, allgather stores the collected data to every participated process. As an opposite of gather, the scatter operation divides the data in the root process and sends the split data to all the processes. The all-to-all operation acts as the ``allscatter", which collectively scatters data on each process to each other.

\subsubsection{Collective computation}
Allreduce/reduce are two popular collective computation operations. We use MPI\_SUM as an example to explain the working principle as it is frequently used. The reduce routine will sum up all the data entries from all the processes in the same communicator and store the sum into the root process. The Allreduce does the same thing but keeps a copy of the sum on every process in that communicator. Another widely-used operation is the reduce-scatter, which acts like a combination of the reduce and scatter: the reduced sum is scattered into all the processes.

\subsection{High-speed Lossy Compressors}

Compression developers and scientific researchers have shown significant interest in high-speed lossy compression due to its ability to achieve high compression ratios. SZ~\cite{Di2016SZ, Tao2017SZ, Liang2018SZ} is an example of a fast, error-bounded lossy compressor, with performance comparable to other compressors like FPZIP~\cite{Lindstrom2006FPZIP} and SZauto~\cite{Zhao2020SZauto}. ZFP~\cite{Lindstrom2014ZFP} is another well-known compressor, offering relatively high compression ratios and even faster speeds than SZ. However, neither can match the speed of SZx~\cite{Yu2022SZx}, which achieves a compression throughput of 700-900 MB/s on CPUs~\cite{Yu2022SZx}. Recently, an ultra-fast compressor, {\fzlight}, has been proposed, but it has yet to be compared with the state-of-the-art SZx compressor used in the C-Coll framework~\cite{huang2023ccoll}. In Section \ref{sec:high-speed-compressor}, we thoroughly compare {\fzlight} with SZx in terms of compression throughput, ratio, and quality, demonstrating that {\fzlight} is the most suitable lossy compressor for MPI collective operations. Accordingly, we develop our customized compressor based on {\fzlight} for MPI collectives, which will be detailed in Section \ref{sec:Customize-fzl-reduce}.

\subsection{Lossy Compression-enabled MPI Implementations}
Researchers have shown interest in using lossy compression to improve MPI communication performance for years. Zhou et al. proposed GPU-compression enhanced point-to-point communication \cite{Zhou2021GPUCOMPRESSION}, and several optimized MPI collective operations \cite{Zhou2022GPUCOMPRESSIONALLTOALL, Zhou2022HiPC} using 1D fixed-rate ZFP~\cite{Lindstrom2014ZFP} on GPUs, but their solutions are either showcasing limited overlapping between compression and communication or subject to the fixed-rate mode compression, leading to the substandard compression quality, as demonstrated in~\cite{huang2023ccoll}. Conversely, our general framework can optimize the collective performance of all MPI collectives while maintaining controlled accuracy.

\section{{\pname} Design and Optimization}
\label{sec-design_and_optimizations}

Figure \ref{fig:architecture} presents the overall design architecture of {\pname}. We highlight the newly designed modules as green boxes. The primary contributions lie in the performance optimization layer and the middleware layer. We carefully characterize the performance of multiple state-of-the-art error-bounded lossy compressors in the context of MPI collectives and select the best-qualified compressor -- {\fzlight}, which will be detailed in Section \ref{sec:high-speed-compressor}. We also propose a series of performance optimization strategies specifically for both of the two collective types (data movement and collective computation), as indicated in the figure. The corresponding details will be discussed in Sections \ref{sec-data-movement-framework}, \ref{sec-collective-computation}, and \ref{sec:Customize-fzl-reduce}. 

\begin{figure}[ht]
    \centering
    {\includegraphics[width=1\linewidth]{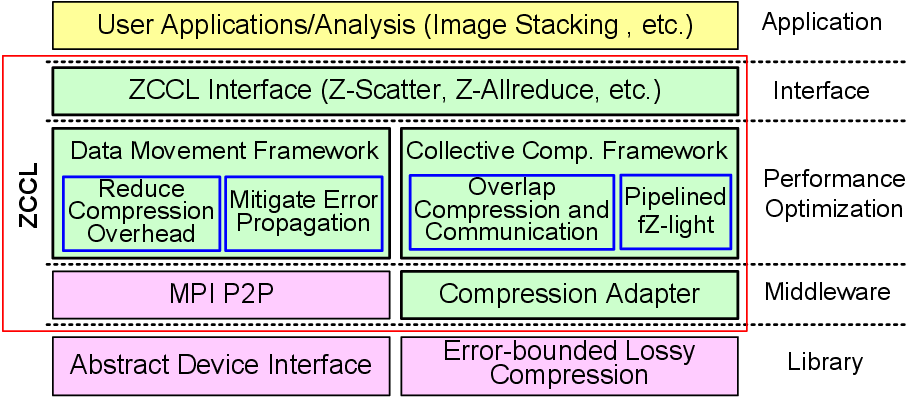}}
    \caption{Design architecture (yellow box: applications; green box: new contributed modules; purple box: third-party).} 
    \label{fig:architecture}
\end{figure}


\subsection{Two Novel Frameworks for Compression-enhanced Collectives}
To integrate lossy compression into MPI collective communication, at least two important aspects must be considered: performance and accuracy. In general, MPI collective operations can be divided into two groups: collective data movements, and collective computation. Instead of directly using the CPRP2P method, we propose two frameworks to implement collective communication for each group, which can maximize the collective operation performance.


\subsubsection{Collective data movement framework}
\label{sec-data-movement-framework}

In this subsection, we detail our strategies for addressing the issues of communication imbalance and compression overhead in our optimized collective data movement framework.
For collective data movement operations, each process in the same communicator needs to communicate with each other to exchange data. If we directly use the CPRP2P method, the sender needs to compress the data every time before sending it, and the receiver is required to decompress the data upon the data arrival. Most of the compression and decompression overheads in CPRP2P, actually, can be avoided by carefully setting the timing of compression operations, as the original data have not been modified during the intensive communication. Besides, the CPRP2P can cause unbalanced communication in that input data on different processes have various compressed data sizes. Such unbalanced communication will slow down the overall collective performance, resulting in a sub-optimal performance. With our framework, we can balance the communication with a fixed pipeline size as the compressed data sizes are decided at the beginning of the intensive communication. In the following text, we illustrate our idea with two examples: a ring-based allgather algorithm and a binomial tree broadcast algorithm. The same philosophy can be easily extended to other collective algorithms in this category.
\begin{figure}[ht]
    \centering
    {\includegraphics[width=0.9\linewidth]{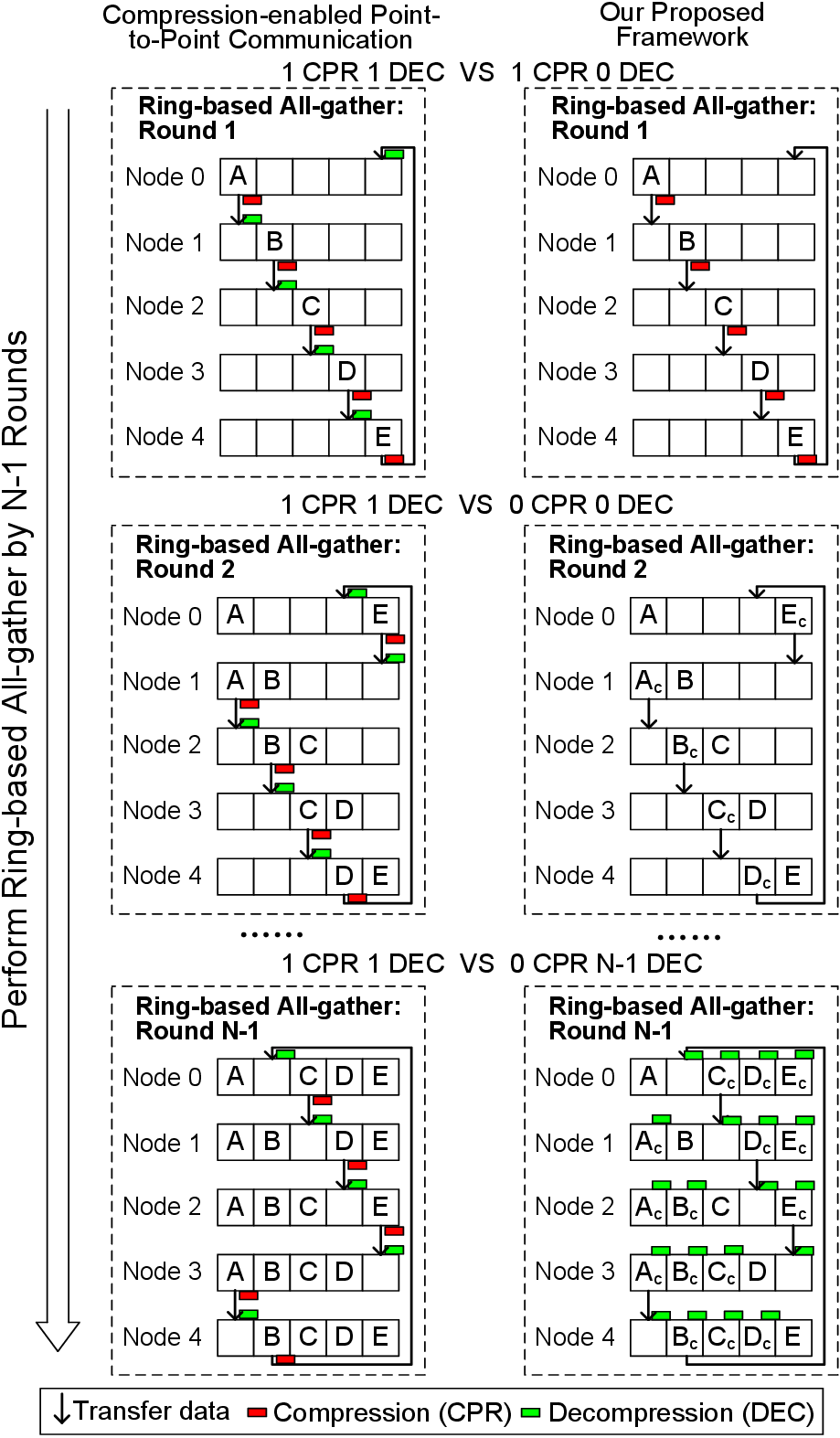}}
    \caption{High-level design of our collective data movement framework in the ring-based allgather algorithm to mitigate compression error propagation. $A$ means the original data and $A_c$ means the compressed data. This rule applies to other data chunks as well. This algorithm completes in $N$$-$1 rounds, where $N$ is the number of processes.} 
    \label{fig-allgather-design}
\end{figure}

Figure \ref{fig-allgather-design} shows the high-level comparison of our proposed framework versus the CPRP2P in the ring-based allgather algorithm. In this algorithm, $N$$-$1 rounds are required to get the gathered results on every process, where $N$ is the number of processes in the communicator. In order to use lossy compression to reduce communication cost in the ring-based algorithm, the straight-forward idea is performing compression and decompression at each round. Instead, our design does not decompress the received data until the last round, thus significantly decreasing the compression overhead from ($N$$-$1)$\cdot$$T_{chunk}$ to $T_{chunk}$, where $T_{chunk}$ is the compression cost of one chunk. When $N$ is large, our novel framework could have nearly $N$$\times$ better performance compared with CPRP2P in terms of compression. Note that the decompression cost remains the same. 

Figure \ref{fig-broadcast-design} presents a similar comparison in the binomial tree broadcast algorithm. There are $log_2 N$ rounds before the completion, where $N$ refers to the total process count. We notice that our proposed framework could reduce the compression and decompression costs from $log_2$$N$$\cdot$($T_{comp}$+$T_{decom}$) to $T_{comp}$+$T_{decom}$ and the performance improvement is $log_2 N$$\times$, where $T_{comp}$ and $T_{decom}$ are the compression time and decompression time of the data at the root process, respectively.

Besides the compression cost, the CPRP2P method can lead to an undesirable error propagation issue during collective sends and receives, as the same data chunk undergoes multiple rounds of compression and decompression. Our framework also solves this issue by compressing the same data chunk for only one time. Similar to the analysis of compression overhead, for the absolute error bounded compression, our proposed framework can decrease the worst case accuracy loss by ($N$$-$1)$\times$ and ($log_2 N$)$\times$ in the ring-based allgather and binomial tree broadcast algorithm, respectively.

\subsubsection{Collective computation framework}
\label{sec-collective-computation}
For collective computation routines, the data entries from all processes in the same communicator need to collectively compute with each other. Unlike in the case of collective data movement, the data transferred in this communication pattern can be updated. As a result, the previous framework cannot be utilized here thus we need to propose a new framework. Despite the updated transferred data precluding us from diminishing compression, we find an opportunity to hide communication inside the compression and decompression.
To clearly elaborate our proposed design, we use the ring-based reduce\_scatter algorithm as an instance. Note that this framework can be easily extended to other collective computation operations. 

\begin{figure}[ht]
    \centering
    {\includegraphics[width=0.9\linewidth]{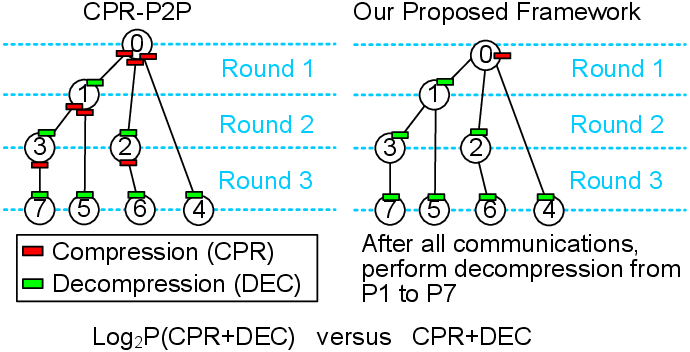}}
    \caption{High-level design of our collective data movement
    framework in the binomial tree broadcast algorithm. It completes in $log_2{N}$ rounds, where $N$ is the number of processes.} 
    \label{fig-broadcast-design}
\end{figure}

Our proposed framework for collective computation is depicted in Figure \ref{fig-reduce-scatter-design}. It employs a ring-based reduce\_scatter algorithm, where each process is required to exchange message chunks with its neighboring processes. For large datasets, these chunk sizes become substantial as they are determined by dividing the size of the input data by the number of processes. In the initial CPRP2P model, compression and decompression occur before and after any communication, respectively. Consequently, a single round incurs three types of overhead: compression/decompression for one message chunk, send/receive operations for the compressed message chunk, and reduction operation for one message chunk. Typically, the compression-related overhead is more significant than the send/receive overhead, as compressed data sizes are considerably smaller than their original counterparts. In our redesigned approach, we significantly mitigate the send/receive overhead by actively pulling communication progress within the compression and decompression phases. This substantially reduces the overall communication time.

\begin{figure}[ht]
    \centering
    {\includegraphics[width=0.9\linewidth]{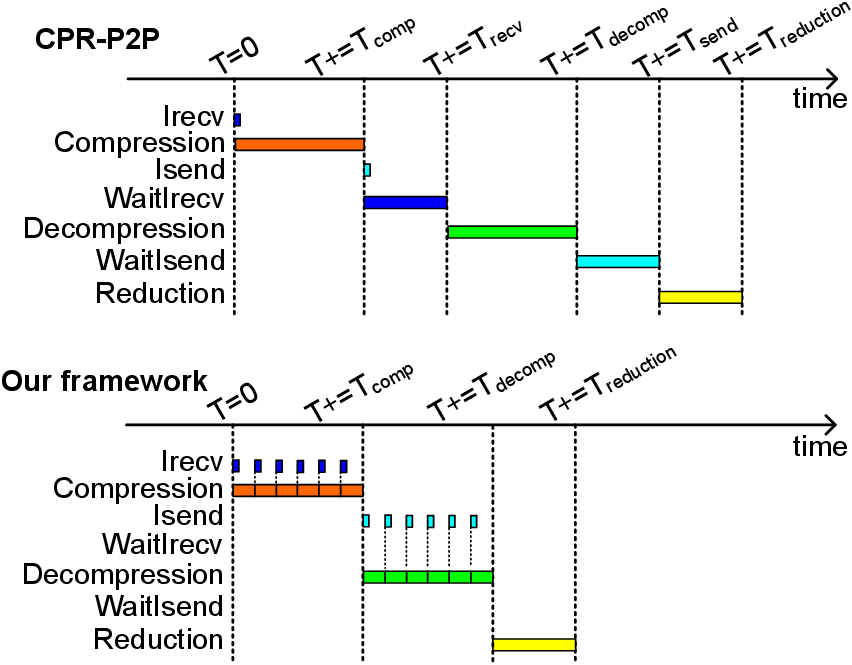}}
    \caption{High-level design of our collective computation framework in the ring-based reduce-scatter algorithm.} 
    \label{fig-reduce-scatter-design}
\end{figure}


\subsection{Theoretical Analysis of Error Propagation in {\pname}}
In this section, we prove the error-bounding nature of our {\pname} framework mathematically. In the following analysis, we assume the lossy compression error $e$ for data $x$ follows a normal distribution, without loss of generality. Specifically, the normal distribution is represented as $e \sim N(\mu, \sigma^{2})$ within the range of $[x- \widehat{e}, x+\widehat{e}]$, where $\widehat{e}$ is the compression error bound. This is exemplified in Figure \ref{fig:normal-distri}, in which we compress climate, weather, seismic wave datasets by SZ3 and ZFP. It clearly shows the normal distribution curve generated by Maximum Likelihood Estimation (MLE) fits the measured compression error values very well for different application datasets.

\vspace{-3mm}
\begin{figure}[ht]
    \centering
        \subfloat[SZ3(Climate)]        {\includegraphics[width=0.275\linewidth]{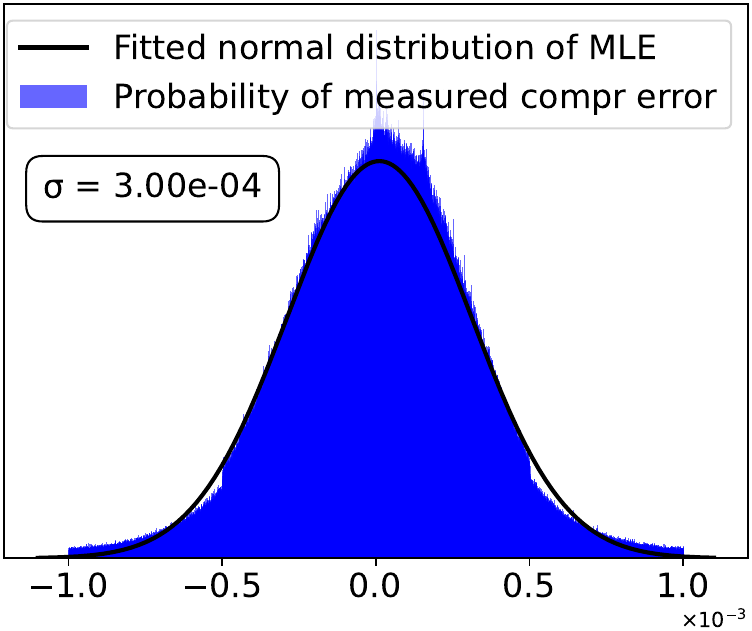}}
        \subfloat[SZ3(Weather)]        {\includegraphics[width=0.35\linewidth]{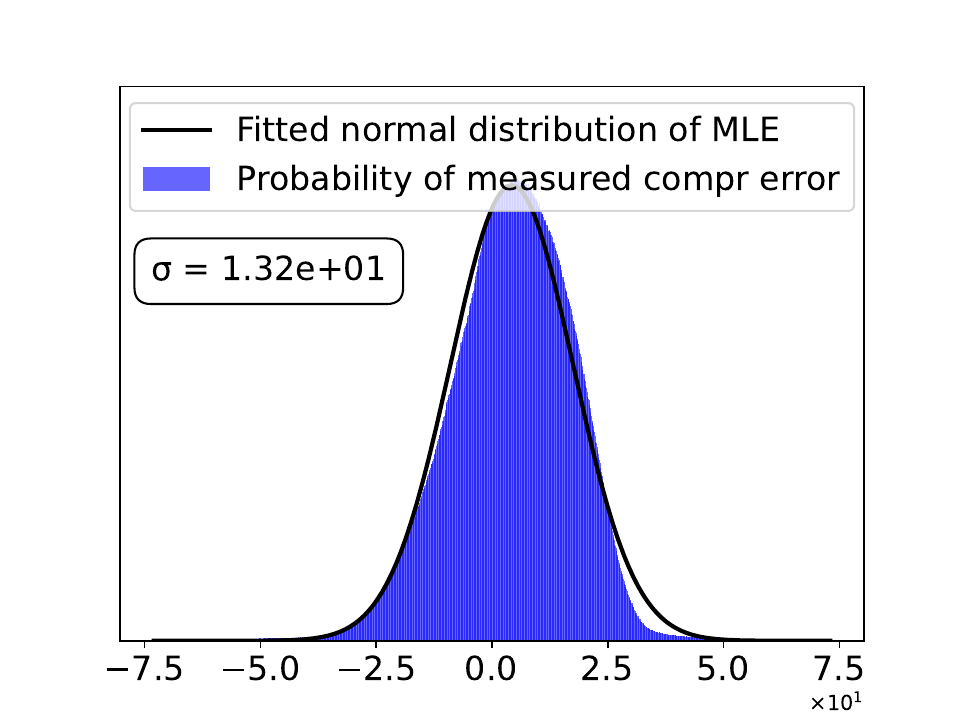}}
        \subfloat[SZ3(Seismic Wave)]   {\includegraphics[width=0.35\linewidth]{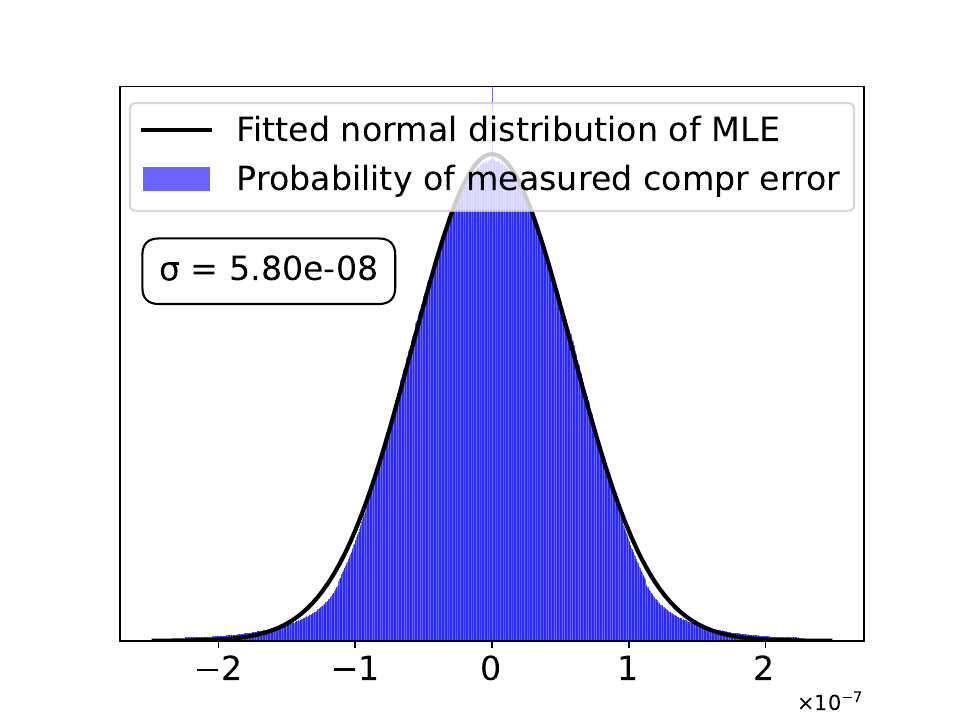}}
        \hspace{3mm}
        \subfloat[ZFP(Climate)]{\includegraphics[width=0.275\linewidth]{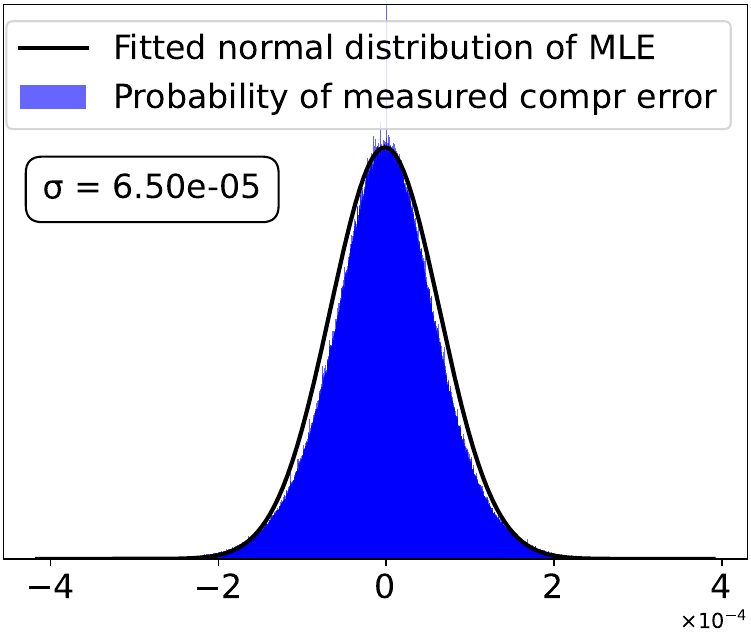}}
        \subfloat[ZFP(Weather)]
        {\includegraphics[width=0.35\linewidth]{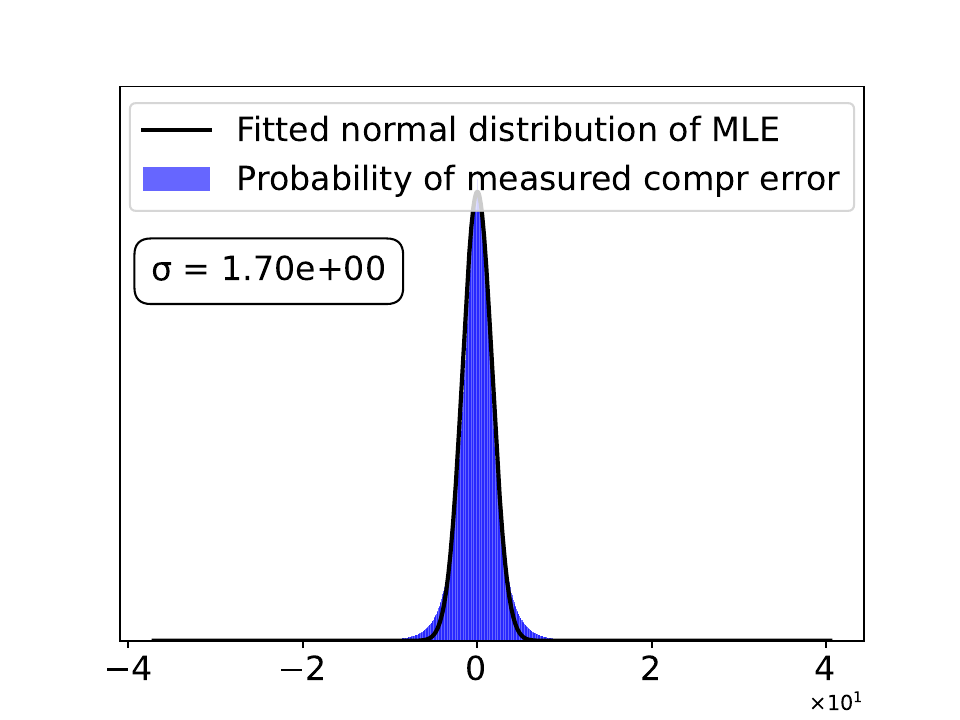}}
        \subfloat[ZFP(Seismic Wave)]
        {\includegraphics[width=0.35\linewidth]{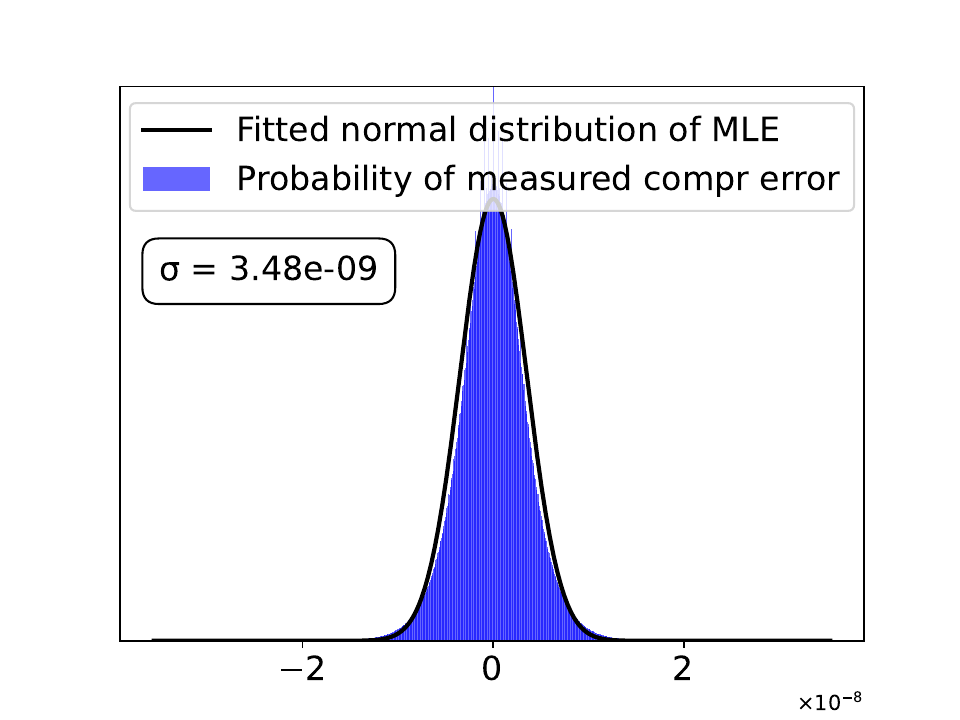}}
        \caption{Exemplifying the normal distribution property of compression errors.}
        \label{fig:normal-distri}
\end{figure}

For the collective communication primitive in MPI, the data will be aggregated gradually from each node during the communication process, which is shown in Fig. \ref{fig-allgather-design}, Fig. \ref{fig-broadcast-design} and Fig. \ref{fig-reduce-scatter-design}. With lossy compression integrated, the compression error will also be aggregated in the data aggregation stage. In the collective data movement framework of {\pname}, each data chunk is compressed only once. Thus, the final error for each data point is within $\widehat{e}$. Unlike the data movement framework, the compression error is aggregated in the collective computation framework of {\pname}, and the aggregation function often involves the \textit{Sum, Average, Max, Min} operations for the floating point data. We further illustrate the error propagation for these operations. For the collective computation framework, we assume there are $n$ data that are collected from $n$ nodes, and the compression error for each of them is $e_{i}$, which follows the normal distribution $e_{i} \sim N(\mu_{i}, \sigma_{i}^{2})$.

\begin{theorem}
Based on the above analysis, the final aggregated error for \textit{Sum} operation falls into the interval $[-2\sqrt{n}\sigma, 2\sqrt{n}\sigma]$ with the probability of $95.44\%$, where $n$ is the number of computing nodes in MPI and $\sigma$ is the variance of the error bound of the lossy compressor.
\end{theorem}

\begin{proof}
The linear combination (e.g., \textit{Sum} in MPI) of normally distributed random variables also follows a normal distribution that is shown in the formula (\ref{norm_distribution}), where $a_{i}$ denotes various constants for $n$ data.

\begin{equation}
\sum\nolimits_{i=0}^{n} a_{i}e_{i} \sim N (\sum\nolimits_{i=0}^{n} a_{i}\mu_{i}, \sum\nolimits_{i=0}^{n} a_{i}^{2} \sigma_{i}^{2})
\label{norm_distribution}
\end{equation}

For the \textit{Sum} operation in the collective computation framework, the compression error will be involved gradually with the aggregation chain, which is shown in the formula (\ref{sum_aggregation}):

\begin{equation}
\begin{aligned}
x_{sum} & = \big(\big(\big((x_{1} + e_{1})+x_{2}\big)+e_{2}\big)+...+e_{n}\big) \\
& = (x_{1}+e_{1})+(x_{2}+e_{2})+...+(x_{n}+e_{n})
\end{aligned}
\label{sum_aggregation}
\end{equation}

Where $x_{i}$ indicates the data collected from node $i$ and $x_{sum}$ represents the final aggregated data. We further demonstrate that the compression error $e_2$ remains to follow the normal distribution after compression through Figure \ref{fig:e2-normal-distri}, which illustrates the probability of the measured compression error ($e_2$) with the fitted MLE curve. The same property also holds for subsequent compression errors in the aggregation chain such as $e_3$, $e_4$, and so on.

\vspace{-3mm}
\begin{figure}[ht]
    \centering
        \subfloat[SZ3($e_2$)]        {\includegraphics[width=0.35\linewidth]{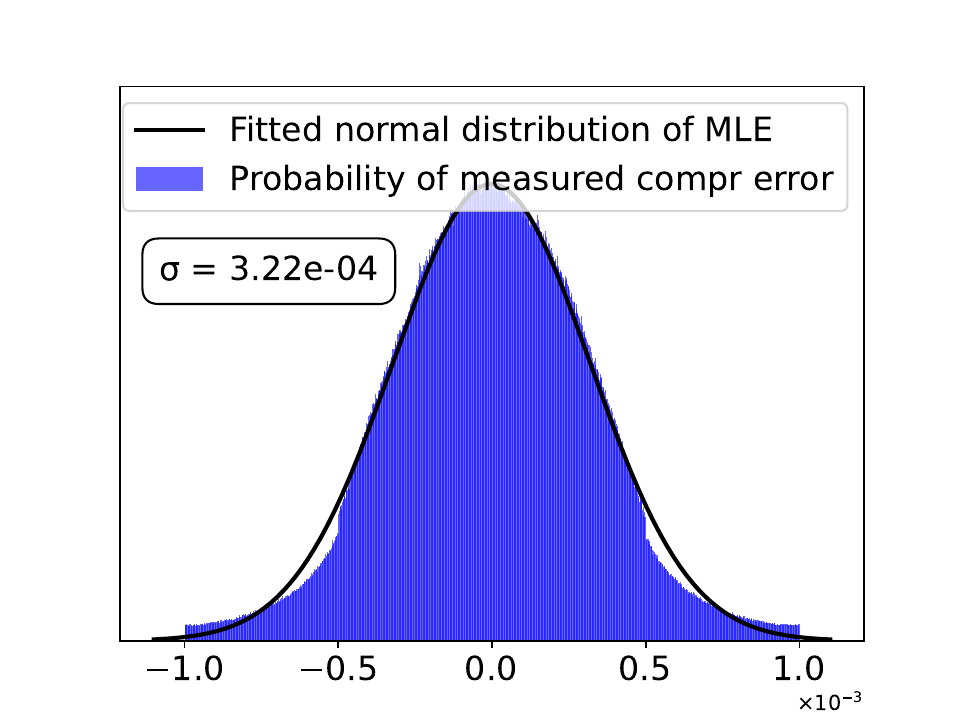}}
        \hspace{3mm}
        \subfloat[ZFP($e_2$)]
        {\includegraphics[width=0.35\linewidth]{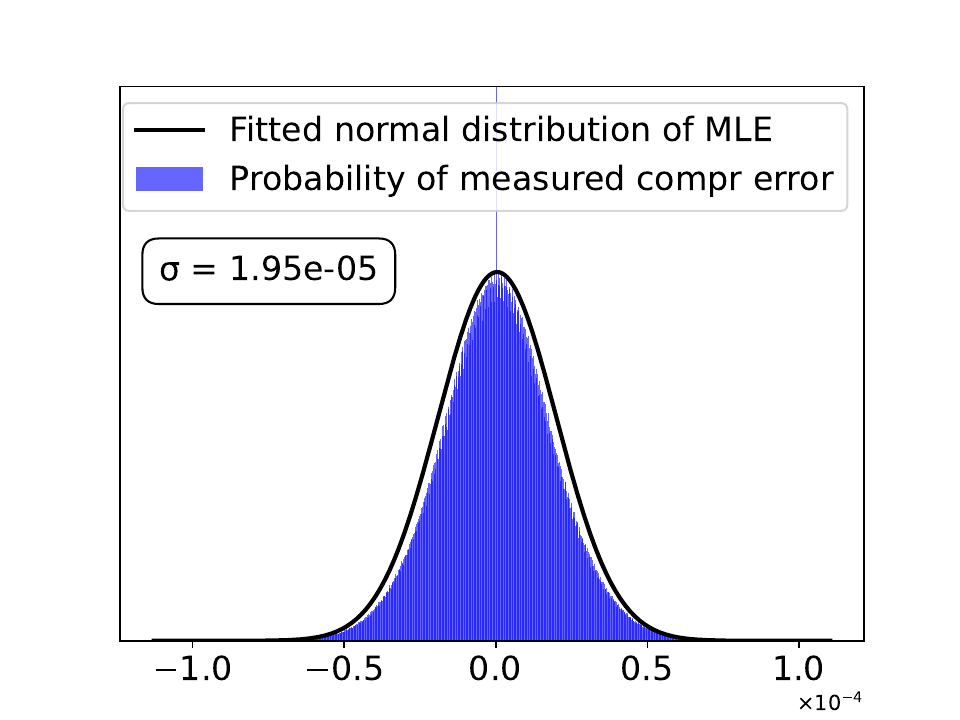}}
        \caption{Exemplifying the normal distribution property of compression error $e_2$.}
        \label{fig:e2-normal-distri}
\end{figure}

From the formula (\ref{sum_aggregation}), we can calculate the aggregated compression error as $\widetilde{e}_{sum} = \sum_{i=0}^{n} e_{i}$. The aggregated error $\widetilde{e}_{sum}$ follows the normal distribution as formula (\ref{error_distribution}):

\begin{equation}
\widetilde{e}_{sum} \sim N(\sum\nolimits_{i=0}^{n} \mu_{i}, \sum\nolimits_{i=0}^{n} \sigma_{i}^{2})
\label{error_distribution}
\end{equation}

The formula (\ref{error_distribution}) indicates that the variance $\sigma^{2}$ of the aggregated error $\widetilde{e}_{sum}$ will be restricted well. When we utilize the same compression error bound across various nodes, the aggregated error conforms to a normal distribution represented as $\widetilde{e}_{sum} \sim N(0, n\sigma^{2})$. Therefore, the final aggregated error falls within the interval $[-2\sqrt{n} \sigma, 2\sqrt{n} \sigma]$ with the probability of $95.44\%$ according to the properties of normal distribution. 

\end{proof}

Since the compression error is bounded by the error bound $\widehat{e}_{i}$ and the error follows the norm distribution, we can assume that $\widehat{e}_{i} \approx 3 \sigma_{i}$ ($\widehat{e}_{i}$ bounded to $3 \sigma_{i}$ with probability of $99.74\%$).

\begin{corollary}
Based on the above assumption, the final aggregated error falls within the interval $[-\frac{2}{3}\sqrt{n} \widehat{e}, \frac{2}{3}\sqrt{n} \widehat{e}]$ with the probability of $95.44\%$. For example, if there are $100$ nodes, the final aggregated error will be bounded within the range $[-\frac{20}{3} \widehat{e}, \frac{20}{3} \widehat{e}]$ with a probability of $95.44\%$.

\end{corollary}

\begin{corollary}
Being similar to the \textit{Sum} operation, the final aggregated error for the \textit{Average} operation in collective computation framework follows the normal distribution $\widehat{e}_{avg} \sim N(0, \frac{\sigma^{2}}{n})$. The final aggregated error will be reduced extremely compared to the original error by $n$ times. 
\end{corollary}

\begin{theorem}
For the \textit{Max, Min} operations, the final error follows the normal distribution $ \widetilde{e}_{max, min} \sim N (0, (2-\frac{n+2}{2^{n}})\sigma^{2})$.
\end{theorem}

\begin{proof}
Since we need to compare the data from the neighbored node gradually, there is a $\frac{1}{2}$ probability that we can choose the non-compressed data. Otherwise, the selected data will contain an error within the error bound $\widehat{e}$. Therefore, the variance of the final aggregated error can be calculated as following formula (\ref{max_min}):
\begin{equation}
\vspace{-3mm}
\frac{1}{2^{n}}n\sigma^{2} + \frac{1}{2^{n-1}}(n-1)\sigma^{2} +...+\frac{1}{2}\sigma^{2} = (2-\frac{n+2}{2^{n}}) \sigma^{2}
\label{max_min}
\end{equation}
\end{proof}

\subsection{Identify Best-qualified High-speed Error-bounded Lossy Compressor}
\label{sec:high-speed-compressor}

In this section, we compare various lossy compressors to identify the most suitable one for MPI collectives. As highlighted in previous analysis, along with controlling data distortion through error bounds, two critical metrics are compression throughput and compression ratio. Prior literature \cite{Di2016SZ,Liang2018SZ,sz3,Zhao2020SZauto,Yu2022SZx} shows that SZx achieves significantly higher compression speed than other compressors, including SZ2 \cite{Liang2018SZ}, SZ3 \cite{sz3}, FPZIP \cite{Lindstrom2006FPZIP}, Auto-SZ \cite{Zhao2020SZauto}, and ZFP \cite{Lindstrom2014ZFP}. Another high-speed compressor, {\fzlight}, is optimized for multi-core CPU architectures, though its performance compared to SZx remains unclear. Therefore, we focus on SZx and {\fzlight} to determine the best compressor for compression-enabled collective communication.

We introduce SZx and {\fzlight} as follows. SZx divides the input data into small blocks (e.g., 128 floating point values) and calculates the mean $\mu$ of the maximum and minimum values for each block. If all data points within a block fall within the interval $(\mu-e,\mu+e)$, where e is a user-defined compression error bound, the block is labeled as a `constant block', and SZx uses the mean $\mu$ to represent the entire block. If some data points fall outside this interval, the block is classified as a `non-constant block', and SZx applies IEEE 754 analysis to compress the data block. These operations primarily involve bitwise operations, addition, or subtraction, making SZx extremely fast. In contrast, {\fzlight} employs a multi-layer block partitioning approach. It first divides the input data into larger thread-blocks, each handled by a single thread, and then further subdivides these thread-blocks into smaller blocks. Next, {\fzlight} performs fused quantization and Lorenzo prediction on each thread-block, converting data values into integers, with the first value stored as an outlier (occupying four bytes). It then obtains the sign-bits and code-length $\log max$, where $max$ is the maximum integer within the small block, for each block. Finally, {\fzlight} applies an ultra-fast bit-shifting encoding scheme to compress the integers. If the code-length is 0, the block is considered to be a `constant block', and only the one-byte code-length is stored in the compressed bytes. Otherwise, the code-length and other fixed-length encoded bits are stored. While {\fzlight}’s operations are also lightweight, the multiplication involved in the quantization stage may lead to higher computational costs compared to SZx.

We compare SZx and {\fzlight} in terms of compression throughput, ratio, and compression quality across four different application datasets. The details of these application datasets are summarized in Table~\ref{tab:application-datasets} in Section~\ref{sec:setup}. All experiments were conducted on a node with two Intel Xeon E5-2695v4 CPU sockets. The performance of compression-enabled collectives is closely tied to the compression throughput and ratio of the chosen compressor, while the data quality of the collective output depends on the compression quality of the selected compressor. For compression throughput, we start with the single-thread performance. As shown in Table \ref{tab:Single-thread-compression}, SZx outperforms {\fzlight} in the RTM application dataset, while {\fzlight} is faster than SZx in most cases of the Hurricane application dataset. For the Nyx and CESM-ATM datasets, SZx and {\fzlight} exhibit comparable compression and decompression throughputs. Overall, SZx and {\fzlight} demonstrate similar compression performance in single-thread mode. Since both SZx and {\fzlight} are optimized for multi-core CPU architectures, we also evaluate their multi-thread compression throughput, as shown in Table \ref{tab:Multi-thread-compression}. In contrast to the single-thread scenario, {\fzlight} consistently outperforms SZx across all relative error bounds and application datasets in multi-thread mode. The superior performance of {\fzlight} in this mode is attributed to its lightweight computational costs and efficient memory access pattern, enabled by the multi-layer block partitioning approach, along with fused quantization and Lorenzo prediction, and the ultra-fast bit-shifting encoding scheme. In summary, {\fzlight} demonstrates significantly better compression performance than SZx in multi-thread mode.

Apart from compression performance, another critical factor influencing collective performance is the compression ratio. In table \ref{tab:compression_ratio_cb_percent}, we present the compression ratios and percentages of constant blocks for SZx and {\fzlight} when compressing four distinct application datasets. We observe that {\fzlight} consistently outperforms SZx in compression ratio across different application datasets. This advantage is due to {\fzlight}'s quantization and Lorenzo prediction stages, which significantly reduce the entropy of the original data, making it easier to compress. In contrast, SZx operates directly on the higher entropy original data. We also notice that within the same application dataset, a smaller relative error bound leads to a lower percentage of constant blocks for both {\fzlight} and SZx. This reduction in the percentage of constant blocks results in a lower compression ratio, as non-constant blocks require more space in the compressed format compared to constant ones. Based on this analysis, we conclude that {\fzlight} achieves a higher compression ratio than SZx.

Apart from performance factors, compression quality is crucial for achieving accurate collective output. Normalized Root Mean Square Error (NRMSE) is a widely used metric for evaluating the accuracy of reconstructed data. In Table \ref{tab:compression_nrmse}, we assess the NRMSE and its standard deviation for SZx and {\fzlight} across four different application datasets under four relative error bounds. The results show that SZx consistently achieves slightly lower NRMSE in all cases. This can be attributed to SZx's approach of using the median value to represent entire constant data blocks, resulting in extremely low variance. However, although SZx demonstrates better numeric compression accuracy than {\fzlight}, this does not necessarily imply superior actual compression quality, warranting further investigation. Peak Signal-to-Noise Ratio (PSNR) is another important metric for evaluating compression accuracy. The rate-distortion graphs in Figure \ref{fig:PSNR-rate} compare the compression bit rate ($32/compression\_ratio$) and PSNR of SZx and {\fzlight}. The results indicate that, for the same bit rate, {\fzlight} achieves higher PSNR than SZx across the RTM, NYX, and Hurricane application datasets. In the CESM-ATM application dataset, {\fzlight} slightly underperforms SZx at very low bit rates but outperforms SZx when the bit rate exceeds 1. This rate-distortion analysis demonstrates that {\fzlight} offers better numeric compression accuracy than SZx at equivalent compression ratios. To further understand the difference in compression quality, we visualize the reconstructed data from both compressors in Figure \ref{fig:vis-szx-vs-fzl}. When compressing the CLOUD field of the CESM dataset to a compression ratio of 8.3, the reconstructed data from SZx displays horizontal stripe artifacts compared to the original data, while {\fzlight} maintains visual quality consistent with the original one. These artifacts in SZx arise because it flattens entire constant data blocks into a single median value, losing the variance within the block. In contrast, {\fzlight} utilizes Lorenzo prediction to preserve data variance, leading to superior compression quality. In conclusion, {\fzlight} demonstrates better overall compression quality than SZx.

Through comprehensive analysis, we conclude that {\fzlight} is generally the better compressor for compression-enabled collective communication. Thus, we decide to customize {\fzlight} for accelerating collective communication. For comparison, we also implemented compression-enabled point-to-point communication-based collectives using the fixed-rate mode and fixed-accuracy mode of ZFP, which serve as baselines in our evaluation.

\begin{table}[ht]
\caption{Single-thread compression throughput (GB/s). The higher throughput is underlined.}
\label{tab:Single-thread-compression}
\resizebox{\columnwidth}{!}{%
\begin{tabular}{@{}cc|cc|cc|cc|cc@{}}
\toprule
\multicolumn{2}{c|}{\textbf{Throughput (GB/s)}} &
  \multicolumn{2}{c|}{\textbf{RTM}} &
  \multicolumn{2}{c|}{\textbf{Nyx}} &
  \multicolumn{2}{c|}{\textbf{CESM-ATM}} &
  \multicolumn{2}{c}{\textbf{Hurricane}} \\ \midrule
                                   & \textbf{REL} & \textbf{COM} & \textbf{DEC} & \textbf{COM} & \textbf{DEC} & \textbf{COM} & \textbf{DEC} & \textbf{COM} & \textbf{DEC} \\ \midrule
\multirow{4}{*}{\textbf{{\fzlight}}} & 1E-1         & 2.97         & 6.25         & 2.87         & 5.89         & 2.46         & 7.73         & {\ul 2.51}   & 5.21         \\
                                   & 1E-2         & 2.80         & 5.80         & {\ul 1.97}   & 3.93         & 1.22         & {\ul 2.75}   & {\ul 1.60}   & {\ul 3.15}   \\
                                   & 1E-3         & 2.68         & 5.47         & {\ul 1.33}   & 2.68         & 0.75         & {\ul 1.58}   & {\ul 1.23}   & {\ul 2.31}   \\
                                   & 1E-4         & 2.61         & 5.39         & {\ul 0.99}   & {\ul 1.83}   & 0.71         & {\ul 1.30}   & {\ul 1.10}   & 1.93         \\ \midrule
\multirow{4}{*}{\textbf{SZx}}      & 1E-1         & {\ul 3.78}   & {\ul 6.98}   & {\ul 3.60}   & {\ul 7.52}   & {\ul 4.77}   & {\ul 11.23}  & 2.46         & {\ul 6.02}   \\
                                   & 1E-2         & {\ul 3.67}   & {\ul 6.61}   & 1.78         & {\ul 4.34}   & {\ul 1.57}   & 2.25         & 1.36         & 2.82         \\
                                   & 1E-3         & {\ul 3.55}   & {\ul 6.26}   & 1.13         & {\ul 2.76}   & {\ul 1.03}   & 1.37         & 1.15         & 2.28         \\
                                   & 1E-4         & {\ul 3.51}   & {\ul 6.22}   & 0.83         & 1.82         & {\ul 0.93}   & 1.23         & 1.05         & {\ul 1.97}   \\ \bottomrule
\end{tabular}%
}
\end{table}

\begin{table}[ht]
\caption{Multi-thread compression throughput (GB/s). The higher throughput is underlined.}
\label{tab:Multi-thread-compression}
\resizebox{\columnwidth}{!}{%
\begin{tabular}{@{}cc|cc|cc|cc|cc@{}}
\toprule
\multicolumn{2}{c|}{\textbf{Throughput (GB/s)}} &
  \multicolumn{2}{c|}{\textbf{RTM}} &
  \multicolumn{2}{c|}{\textbf{Nyx}} &
  \multicolumn{2}{c|}{\textbf{CESM-ATM}} &
  \multicolumn{2}{c}{\textbf{Hurricane}} \\ \midrule
 &
  \textbf{REL} &
  \textbf{COM} &
  \textbf{DEC} &
  \textbf{COM} &
  \textbf{DEC} &
  \textbf{COM} &
  \textbf{DEC} &
  \textbf{COM} &
  \textbf{DEC} \\ \midrule
\multirow{4}{*}{\textbf{{\fzlight}}} &
  1E-1 &
  {\ul 54.10} &
  {\ul 53.46} &
  {\ul 52.13} &
  {\ul 52.39} &
  {\ul 39.50} &
  {\ul 103.65} &
  {\ul 51.38} &
  {\ul 79.34} \\
                              & 1E-2 & {\ul 50.19} & {\ul 52.58} & {\ul 38.42} & {\ul 46.42} & {\ul 19.97} & {\ul 41.91} & {\ul 26.52} & {\ul 47.86} \\
                              & 1E-3 & {\ul 47.36} & {\ul 50.95} & {\ul 28.45} & {\ul 38.33} & {\ul 14.26} & {\ul 28.98} & {\ul 20.15} & {\ul 34.18} \\
                              & 1E-4 & {\ul 44.09} & {\ul 48.26} & {\ul 22.13} & {\ul 31.85} & {\ul 14.61} & {\ul 26.08} & {\ul 18.02} & {\ul 27.79} \\ \midrule
\multirow{4}{*}{\textbf{SZx}} & 1E-1 & 31.90       & 45.97       & 34.20       & 46.04       & 20.38       & 69.87       & 22.61       & 62.04 \\
                              & 1E-2 & 28.77       & 45.09       & 16.82       & 34.88       & 4.97        & 23.88       & 8.04        & 33.25       \\
                              & 1E-3 & 27.06       & 43.19       & 9.16        & 30.75       & 2.97        & 24.16       & 5.51        & 30.13       \\
                              & 1E-4 & 26.99       & 43.52       & 5.93        & 24.66 & 2.78        & 22.76       & 4.97        & 26.54       \\ \bottomrule
\end{tabular}%
}
\end{table}

\begin{table}[]
\caption{Compression ratio and percentage of constant blocks. The higher ratio is underlined.}
\label{tab:compression_ratio_cb_percent}
\resizebox{\columnwidth}{!}{%
\begin{tabular}{@{}cc|cc|cc|cc|cc@{}}
\toprule
\multicolumn{2}{c|}{\textbf{C.B. = Constant Block}} &
  \multicolumn{2}{c|}{\textbf{RTM}} &
  \multicolumn{2}{c|}{\textbf{Nyx}} &
  \multicolumn{2}{c|}{\textbf{CESM-ATM}} &
  \multicolumn{2}{c}{\textbf{Hurricane}} \\ \midrule
 &
  \textbf{REL} &
  \textbf{Ratio} &
  \textbf{C.B.\%} &
  \textbf{Ratio} &
  \textbf{C.B.\%} &
  \textbf{Ratio} &
  \textbf{C.B.\%} &
  \textbf{Ratio} &
  \textbf{C.B.\%} \\ \midrule
\multirow{4}{*}{\textbf{{\fzlight}}} &
  1E-1 &
  {\ul 129.64} &
  98.91\% &
  {\ul 107.83} &
  96.65\% &
  {\ul 69.45} &
  89.30\% &
  {\ul 73.74} &
  90.99\% \\
                              & 1E-2 & {\ul 107.06} & 97.54\% & {\ul 27.00} & 57.44\% & {\ul 21.76} & 45.30\% & {\ul 25.76} & 62.45\% \\
                              & 1E-3 & {\ul 81.04}  & 96.09\% & {\ul 14.97} & 34.64\% & {\ul 12.61} & 19.55\% & {\ul 13.65} & 48.64\% \\
                              & 1E-4 & {\ul 61.51}  & 95.49\% & {\ul 7.81}  & 21.37\% & {\ul 7.18}  & 12.84\% & {\ul 8.12}  & 44.68\% \\ \midrule
\multirow{4}{*}{\textbf{SZx}} & 1E-1 & 82.72        & 97.68\% & 107.58      & 99.02\% & 68.63       & 95.46\% & 59.89       & 93.50\% \\
                              & 1E-2 & 60.14        & 96.03\% & 11.12       & 55.58\% & 9.29        & 48.65\% & 10.61       & 59.72\% \\
                              & 1E-3 & 45.98        & 94.88\% & 5.76        & 37.49\% & 4.39        & 19.57\% & 6.10        & 43.58\% \\
                              & 1E-4 & 37.60        & 94.52\% & 3.62        & 18.55\% & 3.10        & 12.07\% & 4.59        & 38.82\% \\ \bottomrule
\end{tabular}%
}
\end{table}

\begin{table}[]
\caption{NRMSE and its standard deviation. The lower NRMSE is underlined.}
\label{tab:compression_nrmse}
\resizebox{\columnwidth}{!}{%
\begin{tabular}{@{}cc|cc|cc|cc|cc@{}}
\toprule
\multicolumn{2}{c|}{}      & \multicolumn{2}{c|}{\textbf{RTM}} & \multicolumn{2}{c|}{\textbf{Nyx}} & \multicolumn{2}{c|}{\textbf{CESM-ATM}} & \multicolumn{2}{c}{\textbf{Hurricane}} \\ \midrule
                                   & \textbf{REL} & \textbf{NRMSE}   & \textbf{STD}   & \textbf{NRMSE}   & \textbf{STD}   & \textbf{NRMSE}      & \textbf{STD}     & \textbf{NRMSE}      & \textbf{STD}     \\ \midrule
\multirow{4}{*}{\textbf{{\fzlight}}} & 1E-1         & 4.62E-03         & 3E-03          & 2.17E-02         & 2E-02          & 3.74E-02            & 2E-02            & 2.29E-02            & 2E-02            \\
                                   & 1E-2         & 6.41E-04         & 5E-04          & 3.20E-03         & 3E-03          & 4.88E-03            & 1E-03            & 3.07E-03            & 2E-03            \\
                                   & 1E-3         & 8.12E-05         & 7E-05          & 4.03E-04         & 3E-04          & 5.30E-04            & 9E-05            & 3.43E-04            & 2E-04            \\
                                   & 1E-4         & 8.85E-06         & 7E-06          & 4.69E-05         & 2E-05          & 5.36E-05            & 8E-06            & 3.57E-05            & 2E-05            \\ \midrule
\multirow{4}{*}{\textbf{SZx}}      & 1E-1         & {\ul 3.22E-03}   & 2E-03          & {\ul 1.19E-02}   & 1E-02          & {\ul 1.76E-02}      & 5E-03            & {\ul 1.71E-02}      & 6E-03            \\
                                   & 1E-2         & {\ul 3.34E-04}   & 2E-04          & {\ul 1.40E-03}   & 8E-04          & {\ul 2.11E-03}      & 4E-04            & {\ul 1.75E-03}      & 5E-04            \\
                                   & 1E-3         & {\ul 3.29E-05}   & 2E-05          & {\ul 1.89E-04}   & 1E-04          & {\ul 1.62E-04}      & 4E-05            & {\ul 1.37E-04}      & 3E-05            \\
                                   & 1E-4         & {\ul 3.04E-06}   & 2E-06          & {\ul 1.77E-05}   & 6E-06          & {\ul 1.38E-05}      & 3E-06            & {\ul 1.37E-05}      & 3E-06            \\ \bottomrule
\end{tabular}%
}
\end{table}


\begin{figure*}[ht]
\centering
\includegraphics[width=0.8\linewidth, trim={0 3mm 0 1mm}, clip]{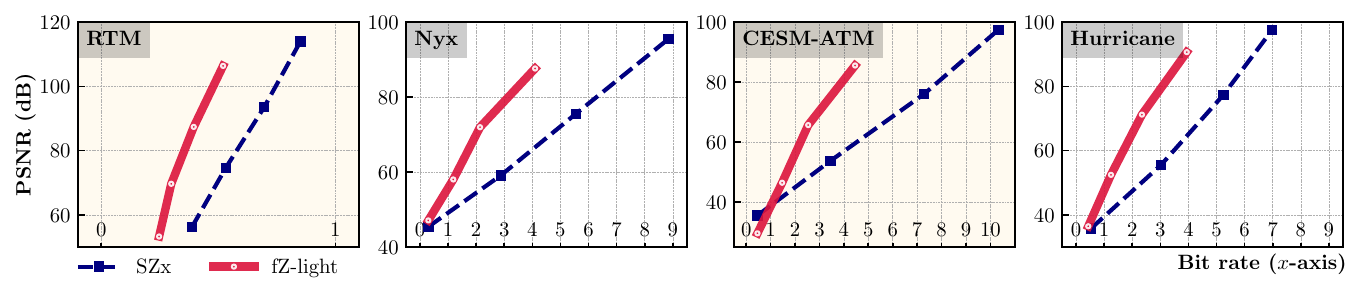}
\caption{Rate-distortion graphs on four application datasets.}
\label{fig:PSNR-rate}
\end{figure*}

\begin{figure}[ht]
    \centering
    \subfloat[Original Data]{\includegraphics[width=0.8\linewidth]{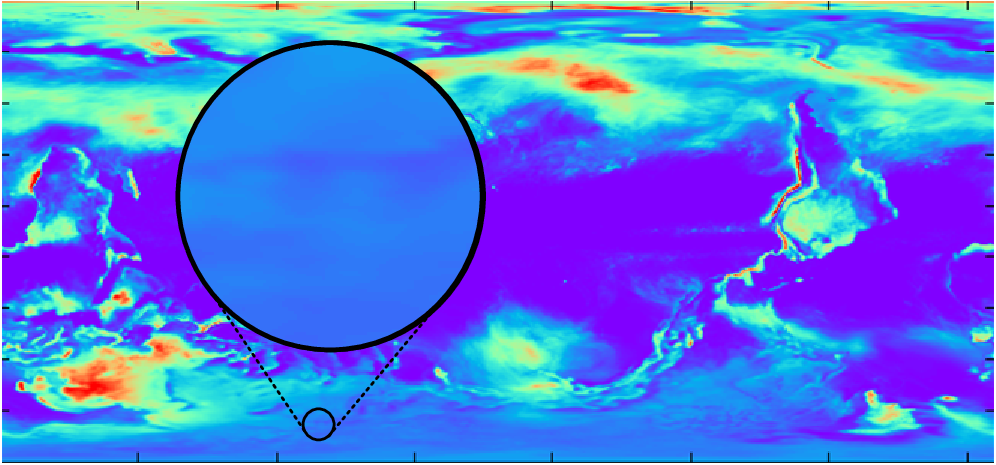}}
    
    \subfloat[SZx]{\includegraphics[width=0.8\linewidth]{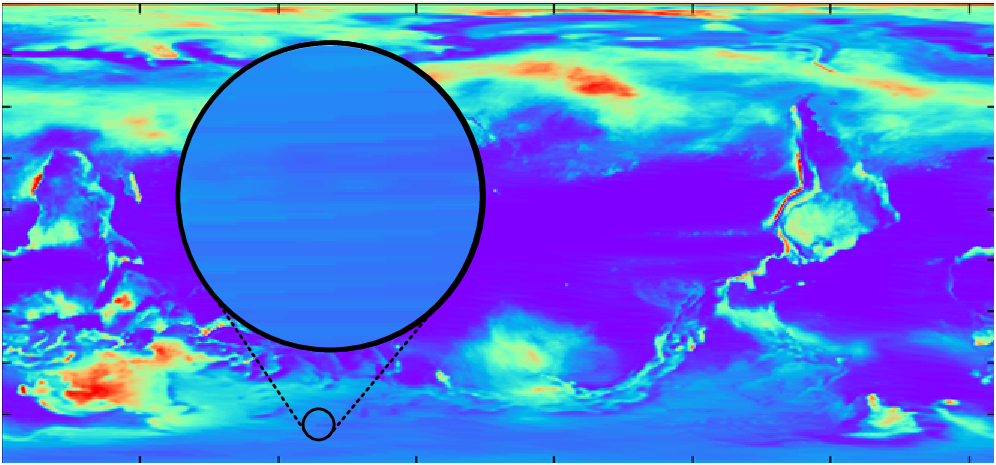}}
    
    \subfloat[{\fzlight}]{\includegraphics[width=0.8\linewidth]{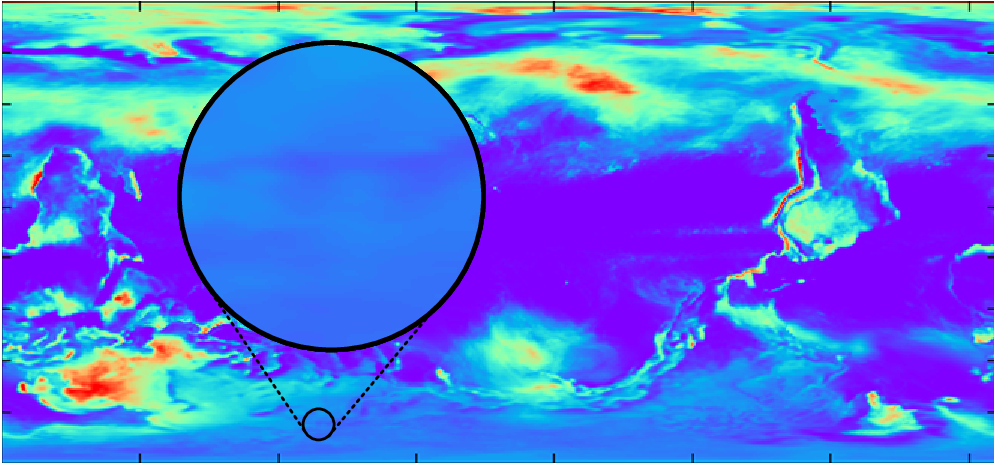}}    

    \caption{Visualization of reconstructed data for SZx vs. {\fzlight} with the same compression ratio of 8.3.}
    \label{fig:vis-szx-vs-fzl}
\end{figure}

\subsection{Characterization of Performance Bottlenecks}
\label{sec:Characterization}


We integrate various high-speed compressors into the point-to-point communication of the ring-based Allreduce algorithm to identify key bottlenecks in collective performance, which will guide our optimization strategies. The ring-based Allreduce is an ideal example as it involves both collective data movement (allgather) and collective computation (reduce\_scatter). Figure \ref{fig-original-naive} in Section \ref{sec:evaluate_different_compression-integrated_baselines} shows the detailed performance breakdown for the direct integration of high-speed compressors using the CPRP2P method. Execution times are normalized to the total runtime of the original MPI\_Allreduce without compression. With {\fzlight} integrated, the CPRP2P method achieves similar overall performance to the original MPI. However, the main bottleneck becomes compression, accounting for 66.42\% of the normalized execution time, followed by communication, which takes 24.58\%. Therefore, optimizing both compression and communication is essential for improving the {\fzlight} integrated baseline.

\subsection{Step-wise Optimizations}

In this subsection, we detail our stepwise optimization strategies, a principal contribution of this paper. To facilitate explanation, we present our implementation of a ring-based Allreduce algorithm accelerated by the {\pname} framework. These optimization strategies are applicable to other collective operations as well. We use {\fzlight} as an exemplar to explain our optimization strategies, which are also applicable to other compressors. Notably, in the ring-based Allreduce, the data transfer required for each process is just $\frac{2(N-1)}{N}$$\cdot$$D_{input}$, where $D_{input}$ represents the input data size and $N$ is the number of processes. Thus, this design is highly efficient for processing long messages. Furthermore, our integration of lossy compression significantly enhances the efficiency of collective communication involving long messages. We have termed our {\pname}-accelerated MPI\_Allreduce as Z-Allreduce, with `Z' denoting compression. Subsequently, we will discuss the design and implementation specifics of Z-Allreduce in a structured, step-by-step approach.

\subsubsection{Utilize our collective data movement framework}
\label{sec:Utilize-data-move}
To reduce the compression overhead and balance communication, we utilize the data movement framework that we presented in Section \ref{sec-data-movement-framework}. At the beginning, every process compresses its local data and stores the compressed data size. Then, every process synchronizes with each other to collect the compressed data sizes in a local integer array $compressed\_sizes$. As the compressed data size only has four bytes, this step is very fast. After that, all processes get the sum of all the compressed data sizes, noted as $total\_count$. Then, each process communicates with each other with a fixed pipeline size until every process has sent $to\_send = total\_count - compressed\_sizes[send\_rank]$ and received $to\_recv = total\_count - compressed\_sizes[self\_rank]$ from other processes in a ring communication pattern. After all communication ends, every process starts to decompress all the received compressed data and store the decompressed data in the receive buffer. Note that they do not need to decompress the data that are compressed by themselves. After this step, we can significantly decrease the time spent by compression and Allgather communication compared with the direct integration of {\fzlight}. Besides, our solution can also preserve the quality/accuracy of the data very well because of the error-bounding feature, which will be demonstrated later in Section \ref{sec-image-stacking}.

\subsubsection{Customize {\fzlight} to reduce communication overhead with our collective computation framework}
\label{sec:Customize-fzl-reduce}
In order to use our collective computation framework in the reduce-scatter stage of the ring-based Allreduce algorithm, we need to redesign the compression workflow of {\fzlight} so that we can consistently poll the progress of the Isend and Irecv inside of the compression and decompression. Therefore, we design and implement the PIPE-{\fzlight} (pipelined {\fzlight}) based on the original {\fzlight}. Instead of compressing the original data as a whole, we divide the compression process into small chunks, each of which handles 5120 data points. Between the compression of two adjacent chunks, we actively poll the communication progress of the non-blocking receive. However, the compressed data of each chunk cannot be simply combined together, otherwise the compressed data cannot be correctly decompressed because each compressed chunk is of variable uncertain length. To solve this problem, we decide to store the compressed data of all chunks in the same output buffer and pre-allocate enough memory space (four bytes per chunk, small memory consumption) at the front of the buffer for storing the compressed data sizes of those chunks together (essentially a kind of index), instead of storing them along with the compressed data chunks. Such a design is more cache-friendly, thus having lower overhead. During the decompression, we maintain a chunk-starting-location pointer based on the recorded compressed chunk sizes to tell the algorithm where the decompression operation should start for each chunk. We repeat this process chunk by chunk and poll the progress of the non-blocking send between decompression chunks. Through this optimization, we can hide the communication in the reduce-scatter stage inside of compression, which further improves the performance of our Z-Allreduce design.

\section{Experimental Evaluation}\label{exp-setup-sec}

In this section, we present and discuss the evaluation results.


\subsection{Experimental Setup}
\label{sec:setup}

Since inter-node communication is the major bottleneck for collectives as discussed previously, we utilized a 128-node cluster with one process per node in our experiments. Each node is equipped with two Intel Xeon E5-2695v4 Broadwell processors. Furthermore, each NUMA node contains 64 GB of DDR4 memory, resulting in a total of 128 GB of memory per node. The nodes are interconnected via Intel Omni-Path Architecture (OPA), providing a maximum message rate of 97 million per second and a bandwidth of 100 Gbps.

\begin{table}[h]
\centering
\caption{Information of the application datasets.}
\label{tab:application-datasets}
\resizebox{\columnwidth}{!}{%
\begin{tabular}{ccccc}
\hline
\textbf{Application} & \textbf{\# fields} & \textbf{Dims per field} & \textbf{Total Size} & \textbf{Domain} \\ \hline
\textbf{RTM~\cite{Kayum2020RTM}} & 151 & 849x849x235 & 95.3 GB & Seismic Wave \\
\textbf{NYX~\cite{nyx}} & 6 & 512x512x512 & 3.1 GB & Cosmology \\
\textbf{CESM-ATM~\cite{hurricane}} & 79 & 1800x3600 & 2.0 GB & Climate Simu. \\
\textbf{Hurricane~\cite{cesm}} & 13 & 100x500x500 & 1.3 GB & Weather Simu.\\\bottomrule 
\end{tabular}%
}
\end{table}

MPI collectives are common operations used in simulation analysis. For instance, generating stacking images in reverse time migration (essentially an Allreduce operation) is a typical real-world use case~\cite{Gurhem2021Kirchhoff}, which will be demonstrated in Section~\ref{sec-image-stacking}. We use the RTM application dataset for our evaluation, as it is the largest dataset we have, as shown in Table\ref{tab:application-datasets}. The information of our solutions and baselines are presented in Table \ref{tab:collective-implementations}. The compression error bound is set to 1E-4 by default. In our experiments, we adopt a two-stage approach, consisting of a warm-up stage and an execution stage. Each stage is run 10 times, and we report the average results of the execution stage to present the overall performance.

\begin{table}[ht]
\centering
\caption{Collective communication solutions.}
\label{tab:collective-implementations}
\resizebox{\columnwidth}{!}{%
\begin{tabular}{@{}ll@{}}
\toprule
\textbf{Solution} & \textbf{Description} \\ \midrule
\textbf{MPI (baseline)} & Orginal MPI collectives with no compression\\
\textbf{CPRP2P (baseline)} & Collectives implemented by CPRP2P with fZ-light\\
\textbf{C-Coll (baseline)} & The current SOTA compression-accelerated collectives~\cite{huang2023ccoll}\\
\textbf{{\pname} (single-thread)} & Single-thread mode of {\pname} \\
\textbf{{\pname} (multi-thread)} & Multi-thread mode of {\pname} \\ \bottomrule
\end{tabular}%
}
\end{table}

\subsection{Evaluating different compression-integrated baselines}
\label{sec:evaluate_different_compression-integrated_baselines}
We compare different compression-integrated CPRP2P baselines with the original MPI without compression using the Allreduce operation across 64 Broadwell nodes in Figure \ref{fig-original-naive}. The execution time of all the baselines is normalized based the running time of the original MPI\_Allreduce. We can notice that fZ-light has the highest performance in all CPRP2P baselines. This is because that fZ-light has the best compression throughput and compression ratio among all its counterparts. It is worth noting that, both the ZFP(ABS) and ZFP(FXR) demonstrate considerably worse performance than SZx and {\fzlight} when integrated into collective communication because of their relatively low compression throughput and ratio as shown in \cite{huang2023ccoll}. Compared to the original MPI, {\fzlight} integrated baseline shows much less communication time due to the significantly decreased communication volume by the compression technique. Since the {\fzlight} integrated baseline presents the highest performance in all the compression-enabled baselines, we use CPRP2P to represent it and compare our ZCCL with it in later experiments.       

\begin{figure}[ht]
    \centering
    {\includegraphics[width=0.8\linewidth]{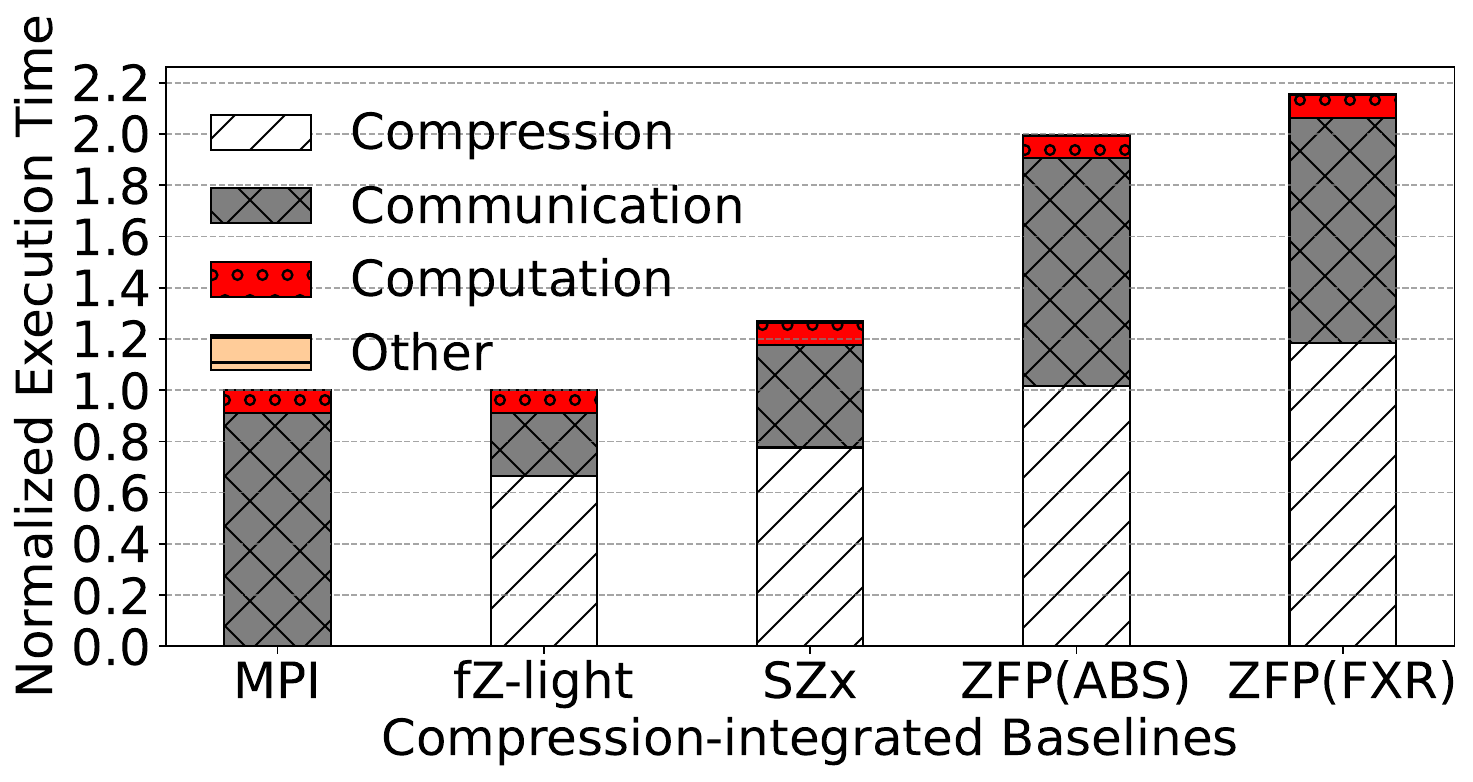}}
    \caption{Compare the normalized execution time of original MPI and the CPRP2P baselines with different compressors.} 
    \label{fig-original-naive}
\end{figure}
\subsection{Step-wise Optimizations to Z-Allreduce with Performance Analysis}
\label{sec:evaluation}
In this section, we carry out step-by-step optimizations to our Z-Allreduce ({\pname}-enhanced Allreduce) integrated with {\fzlight} and demonstrate the performance on 64 Broadwell nodes. These optimizations are also applicable to Z-Allreduce when integrated with other compressors. All compression and decompression operations in this section are conducted in single-thread mode. The ring-based Allreduce that we implement contains a Reduce-scatter stage and an Allgather stage. Thus, we breakdown the performance of the two stages separately.

\subsubsection{Evaluating our collective data movement framework with Allgather}
\label{sec:eval-framework}

Figure \ref{fig-new-design} shows the performance improvement achieved by our novel design in the Allgather stage for data sizes ranging from 50 MB to 600 MB. In {\pname}, our collective data movement framework brings a considerable reduction in both compression and decompression time. Notably, at the data size of 300 MB, our {\pname} achieves a compression speed-up of 3.74$\times$ when compared to CPRP2P approach. Additionally, our balanced communication in {\pname} is up to 1.46$\times$ faster than the unbalanced communication in CPRP2P at 600MB. Overall, our {\pname} achieves the maximum speedup of 3.26$\times$ compared to CPRP2P at 250 MB. We analyze the key reason why our solution can obtain a significant performance improvement as follows. In fact, to overcome the compression bottleneck and balance MPI communication, we utilize our collective data movement framework that pre-compresses the data before transmission and decompresses it after all communication, rather than using expensive CPRP2P in collective routines. This novel design can significantly reduce the amount of compression required during collective communication. Using CPRP2P also brings unbalanced communication as the compressed data sizes may vary, but we can balance the communication with a fixed pipeline size in our new design because we do not need to compress the data every time before we send it. 
\begin{figure}[ht]
    \centering
    {\includegraphics[width=0.99\linewidth]{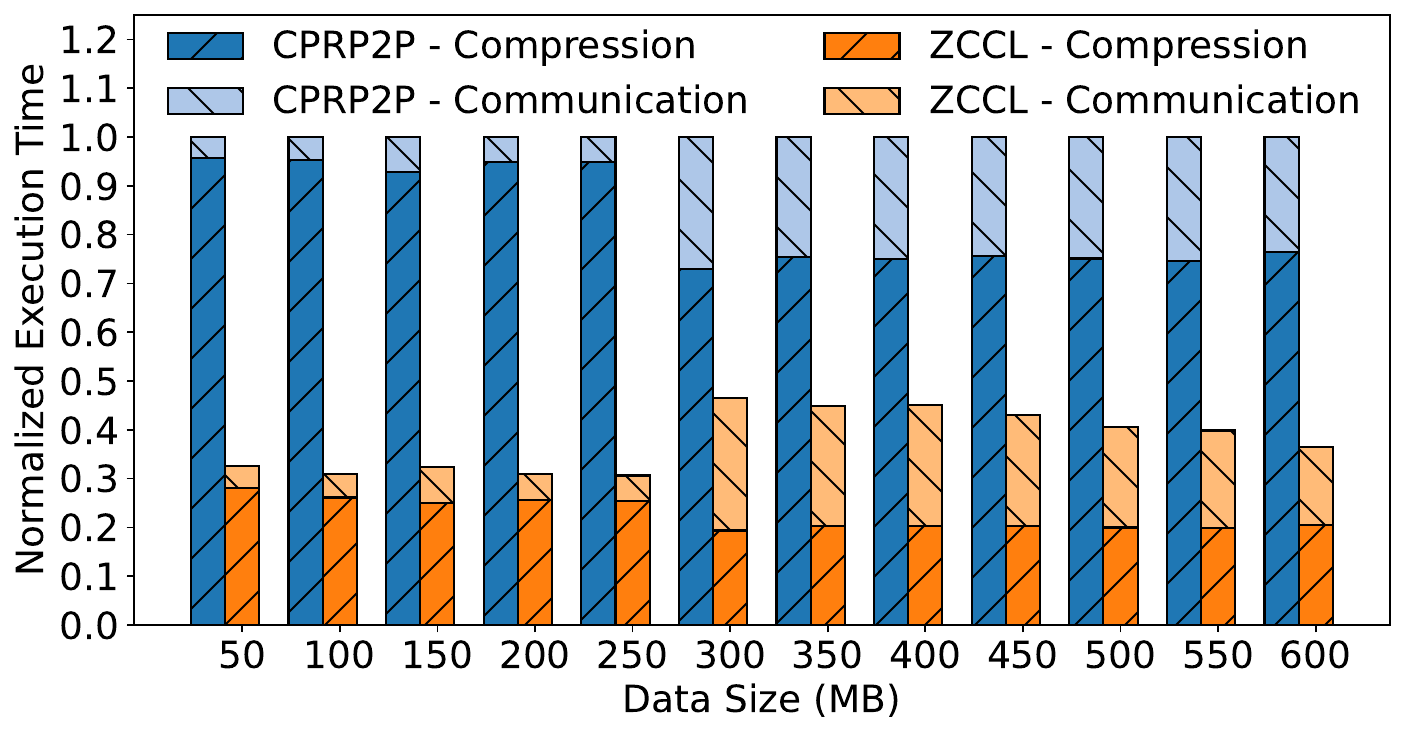}}
    \caption{Compare the Allgather performance of CPRP2P and our {\pname} from 50 MB to 600 MB.} 
    \label{fig-new-design}
\end{figure}

Note that using CPRP2P may accumulate errors during intensive collective communication, such as ring-based communication, as the same data is repeatedly passed from one process to another. Therefore, we have utilized our new framework to ensure that errors in the final results are bounded, which we will discuss in the application evaluation section \ref{sec-image-stacking}.

\subsubsection{Evaluating reduced communication overhead with our collective computation framework}
We demonstrate the effectiveness of our collective computation framework in this section. From Figure \ref{fig-Overlap}, it is evident that our {\pname} leads to significantly less communication in the Reduce\_scatter stage compared with the CPRP2P method, resulting in a performance boost of up to 3.32$\times$ for the data size of 300 MB. The rationale for this performance improvement is shown in the following text. To utilize our collective computation framework for hiding communication during compression, we design and implement PIPE-{\fzlight} (pipelined fZ-light), which could break the compression process into small chunks and allow us to overlap the compression with communication in a fine-grained pipelined manner. As a result, we can significantly reduce communication time in the Reduce\_scatter stage. Combining this optimized Reduce\_scatter with the previously optimized Allgather, we have obtained the final version of our Z-Allreduce and will evaluate it in the following parts.

\begin{figure}[ht]
    \centering
    {\includegraphics[width=0.99\linewidth]{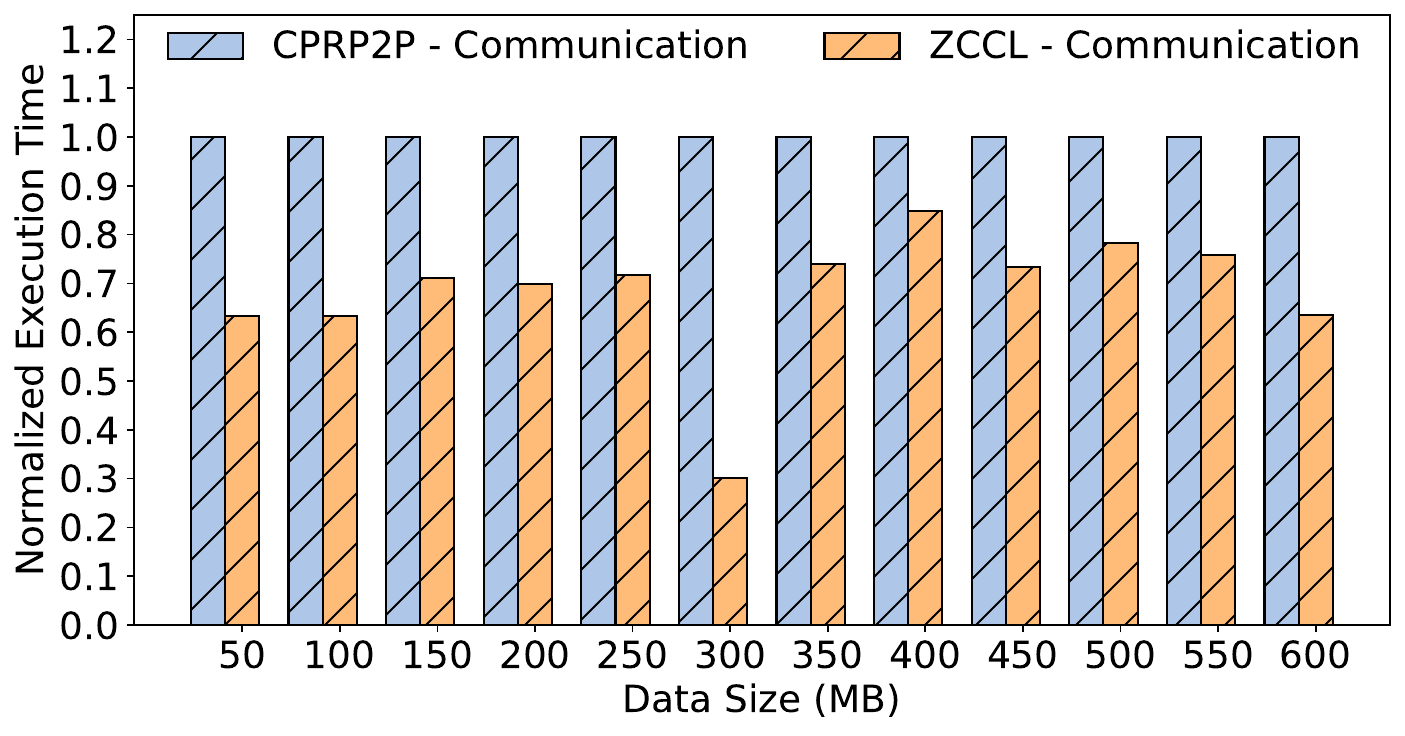}}
    \caption{Compare the Reduce\_scatter communication time of CPRP2P and {\pname} from 50 MB to 600 MB.} 
    \label{fig-Overlap}
\end{figure}

\subsection{End-to-end Comparisons of Z-Allreduce with Baselines}

In this section, we compare the performance of our {\pname}-accelerated Z-Allreduce with four different baselines on various data sizes, node numbers, and datasets.

\begin{figure}[ht]
    \centering
    {\includegraphics[width=0.99\linewidth]{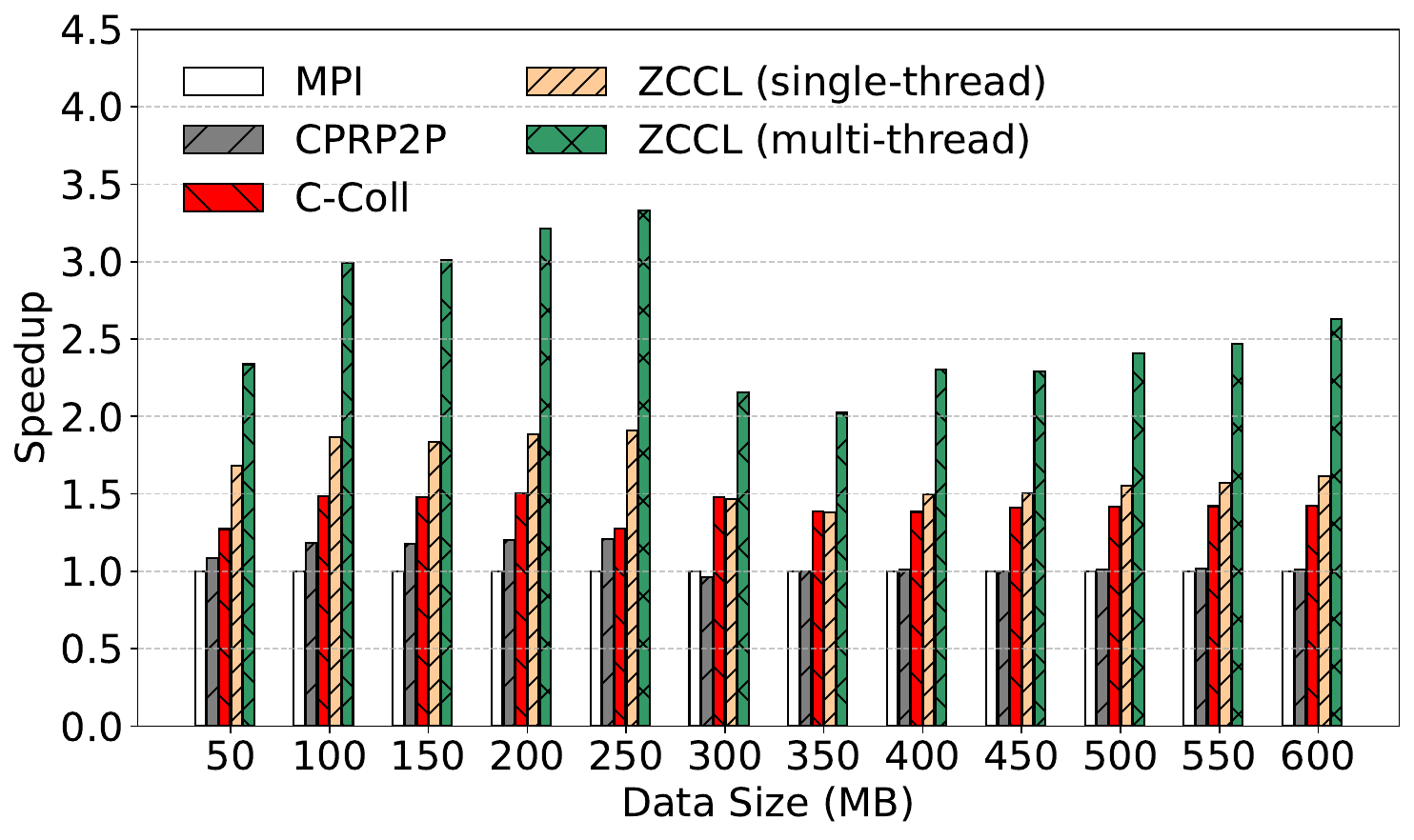}}
    \caption{Compare the performance of our {\pname}-accelerated Z-Allreduce and multiple baselines from 50 MB to 600 MB.} 
    \label{fig-64-sizes}
\end{figure}


\subsubsection{Evaluating with different data sizes}\label{sec-128-sizes}

In this section, we present the performance of our Z-Allreduce, along with related baselines, using data sizes ranging from 50 MB to 600 MB on 64 Broadwell nodes. As shown in Figure \ref{fig-64-sizes}, {\pname} consistently outperforms or matches the state-of-the-art C-Coll, achieving up to 1.50$\times$ and 2.69$\times$ performance improvements in single-thread and multi-thread modes, respectively. Compared with the CPRP2P method, {\pname} achieves even greater performance enhancements, with up to 1.64$\times$ and 2.88$\times$ speedups in single-thread and multi-thread modes. Additionally, {\pname} is up to 1.91$\times$ and 3.46$\times$ faster than MPI in the two modes. This high performance of {\pname} stems from reduced compression and communication overhead, thanks to our collective data movement and computation frameworks. Moreover, we carefully select and customize the {\fzlight} compressor, leading to higher performance than the SZx-integrated baseline C-Coll.

\begin{figure}[ht]
    \centering
    {\includegraphics[width=0.99\linewidth]{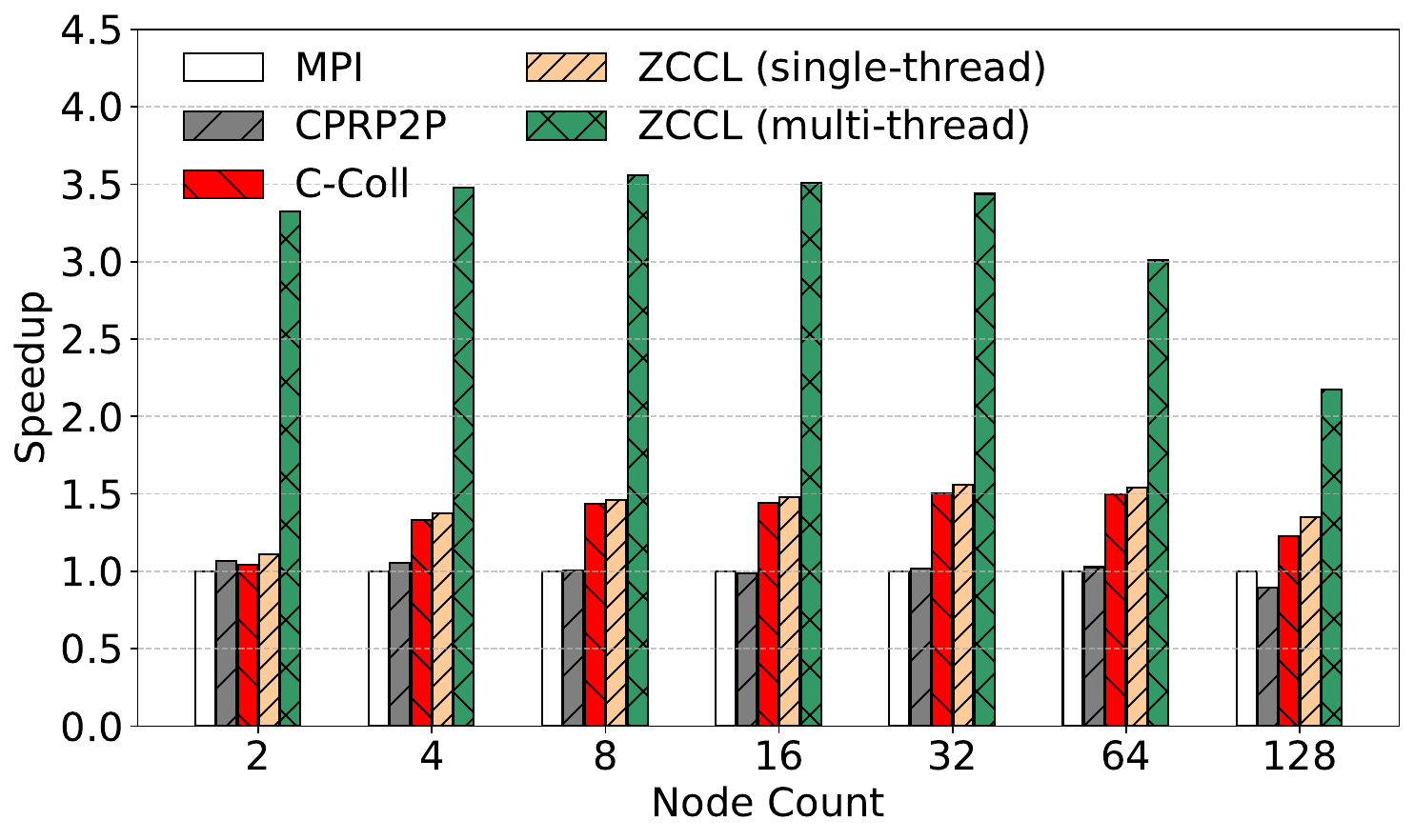}}
    \caption{Compare the performance of our {\pname}-accelerated Z-Allreduce and multiple baselines from 2 to 128 nodes.} 
    \label{fig-128-nodes}
\end{figure}

\subsubsection{Evaluating with different node counts}

To demonstrate the scalability of our approach, we compare the normalized execution time of our Z-Allreduce and four different baselines using a fixed data size of 678MB (the whole RTM dataset) across 2 to 128 nodes. As shown in Figure \ref{fig-128-nodes}, our {\pname} outperforms all the baselines across various node numbers. It can reach performance boosts of up to 1.56$\times$ and 3.56$\times$ in the single-thread and multi-thread versions compared to the MPI, respectively. Similar to our observation in Section \ref{sec-128-sizes}, we found that both the single-thread and multi-thread modes of {\pname} are outperforming the C-Coll framework, achieving 1.1$\times$ and 3.19$\times$ maximal performance improvements, respectively. Compared with the CPRP2P baseline, our {\pname} even reaches better speedups, with up-to 1.23$\times$ in the single-thread mode and 2.99$\times$ in the multi-thread mode. This scalability evaluation again confirms the high-performance of our designs and optimizations in the {\pname} framework.

\subsection{Generalizability Demonstration on Other MPI Collectives}

We have demonstrated the high performance of our {\pname}-accelerated Z-Allreduce, consisting of Z-Allgather and Z-Reduce-scatter. To showcase the generalizability of our frameworks and optimizations, we also present Z-Bcast and Z-Scatter, which utilize the ubiquitous binomial tree algorithm adopted by MPICH. We conduct experiments ranging from 50 MB to 600 MB using 64 Broadwell nodes.
\subsubsection{Broadcast}\label{sec-eval-bcast}

In Figure \ref{fig-portability-bcast}, we present the speedups of our Z-Bcast normalized against the original MPI\_Bcast. We also compare Z-Bcast with the C-Coll baseline. The experimental results show that {\pname} is 1.6$\times$ and 8.9$\times$ faster than MPI in single-thread and multi-thread modes, respectively. These performance improvements are originated from the reduced data transfer volume and minimized compression overheads provided by our framework. Additionally, {\pname} surpasses C-Coll, achieving up to 1.1$\times$ and 7.5$\times$ speedups in single-thread and multi-thread modes, respectively, demonstrating superior communication efficiency.
\begin{figure}[ht]
    \centering
    {\includegraphics[width=0.99\linewidth]{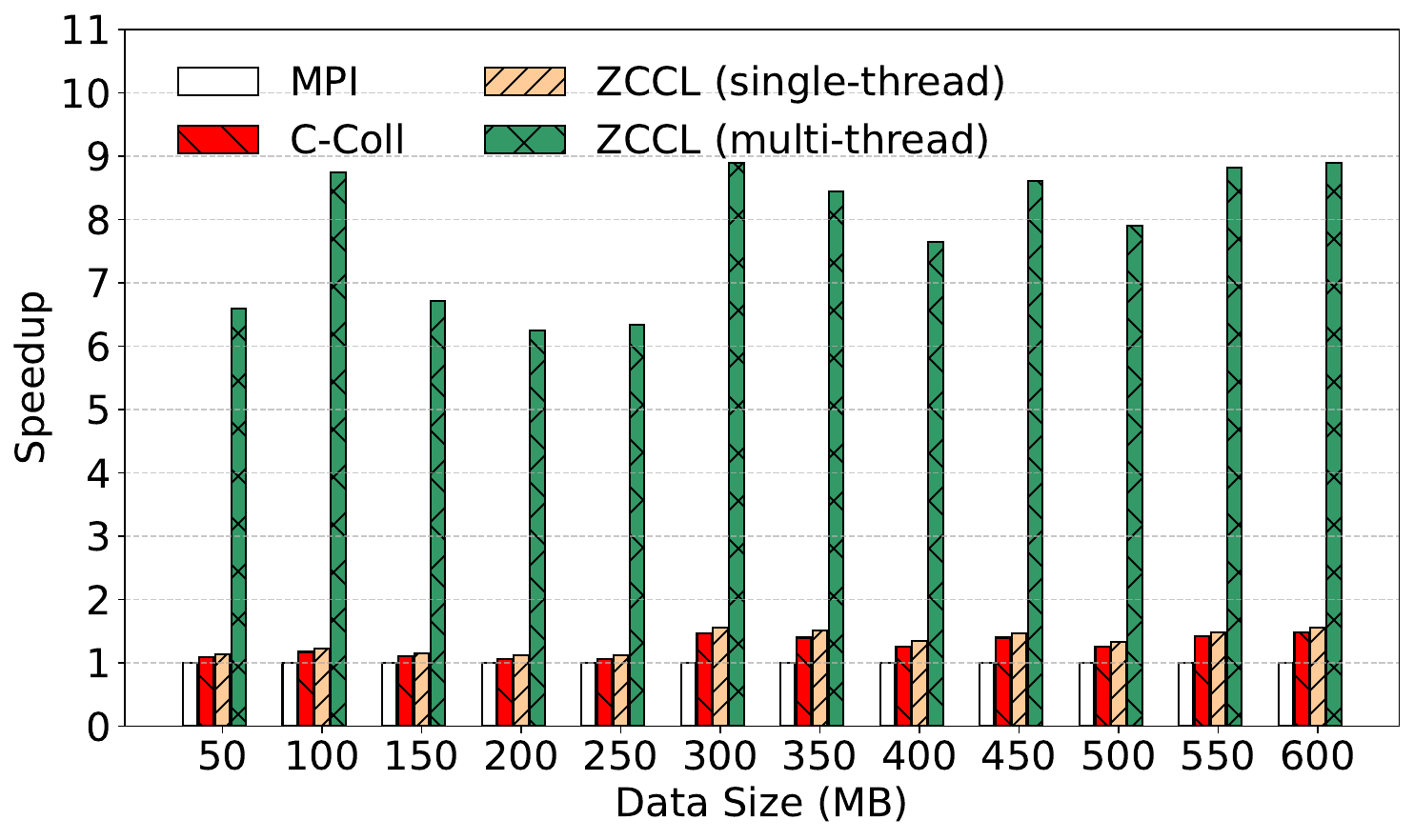}}
    \caption{Generalizability demonstration of our proposed framework and optimizations with Bcast from 50 MB to 600 MB.} 
    \label{fig-portability-bcast}
\end{figure}
\subsubsection{Scatter}\label{sec-eval-scatter}

We evaluate our Z-Scatter against MPI and C-Coll in Figure \ref{fig-portability-scatter}, and observe that {\pname} achieves the highest performance among all counterparts. It demonstrates 1.5$\times$ and 5.4$\times$ performance improvements over MPI, and 1.2$\times$ and 4.5$\times$ enhancements over C-Coll. These speedups are even more pronounced compared to our Z-Allreduce, as collective data movement benefits more from our framework than collective computation.

\begin{figure}[ht]
    \centering
    {\includegraphics[width=0.99\linewidth]{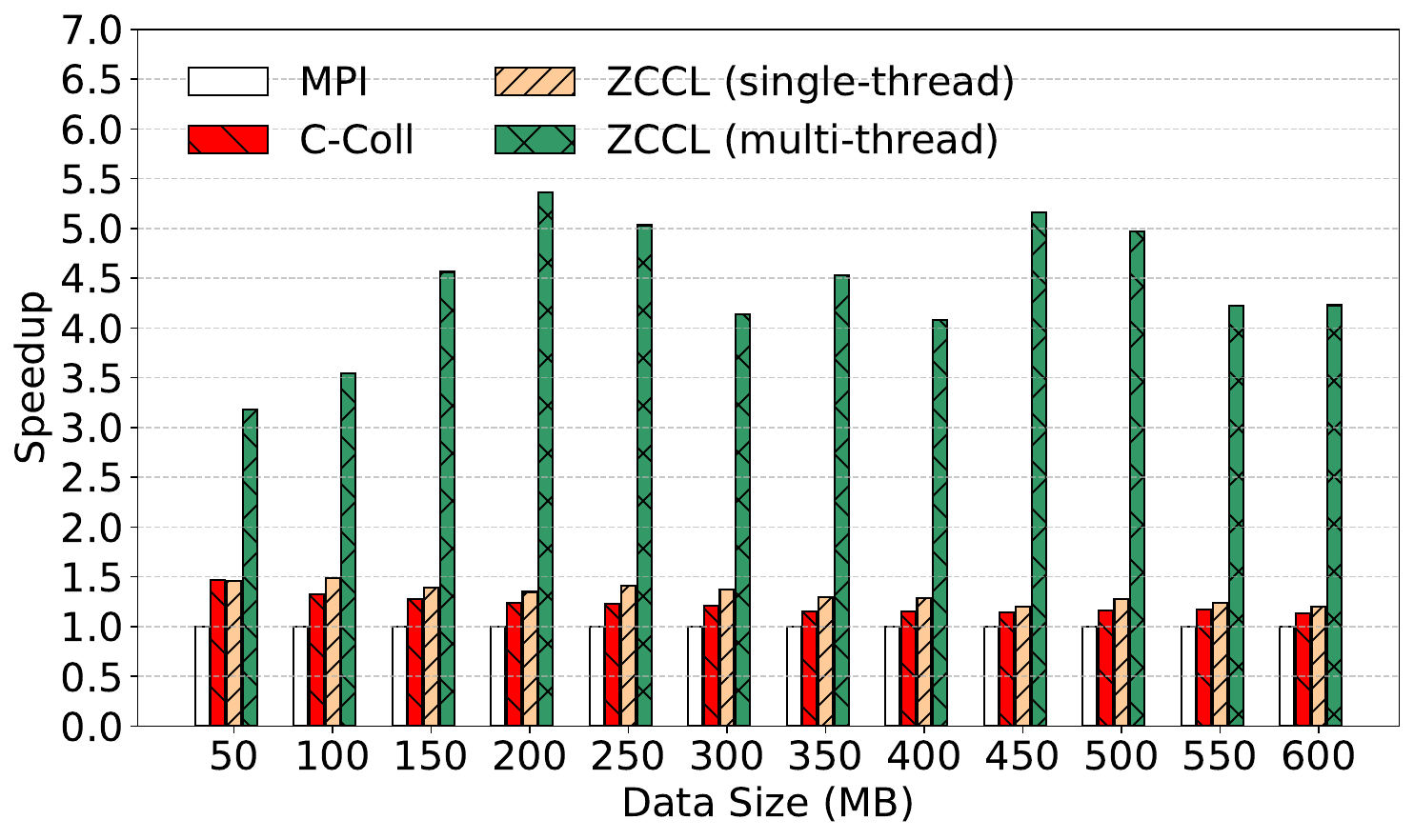}}
    \caption{Generalizability demonstration of our proposed framework and optimizations with Scatter from 50 MB to 600 MB.} 
    \label{fig-portability-scatter}
\end{figure}

\subsection{Evaluation of Image Stacking Performance and Accuracy}
\label{sec-image-stacking}

We use the image stacking application to evaluate both the performance and accuracy of our {\pname}. Image stacking is a widely used technique in scientific domains such as climate simulation and geology to generate high-quality images by combining multiple individual images. Researchers employ MPI to sum these images into final composite images \cite{Gurhem2021Kirchhoff}.

We show the performance results and validate the high quality of the stacked images generated under our compression-accelerated collective framework in the following texts. As shown in Table \ref{tab:image-stacking-perf}, for {\pname}, we observe speedups of 1.61$\times$ and 2.96$\times$ compared to MPI in single-thread and multi-thread modes, respectively. In contrast, the C-Coll baseline achieves only a 1.19$\times$ performance improvement, while the CPRP2P baseline lags behind the original MPI without compression. In CPRP2P, compression constitutes the majority of the overall runtime (63.12\%), highlighting its inefficiency in compression costs. This scenario is improved in the C-Coll framework, where compression takes a smaller proportion (53.47\%), and communication accounts for 34.24\%. In the single-thread mode of {\pname}, compression remains the largest contributor to runtime at 58.23\%, but this is still a sound improvement over CPRP2P, thanks to the significantly reduced compression time in {\pname}. Additionally, communication time decreases substantially compared to C-Coll, due to the higher compression ratio of {\fzlight} compared to SZx. In the multi-thread mode of {\pname}, compression costs are further reduced to only 23.18\%, with communication accounting for the largest proportion at 50.65\%. This is a result of the superior multi-thread compression performance of our {\pname} framework, boosted by the {\fzlight} compressor.

\begin{table}[]
\caption{Performance comparison and breakdown of image stacking (The speedup is based on MPI. The last four columns are performance breakdowns).}
\label{tab:image-stacking-perf}
\resizebox{\columnwidth}{!}{%
\begin{tabular}{@{}c|c|cccc@{}}
\toprule
                              & \textbf{Speedup} & \textbf{Compre.} & \textbf{Commu.} & \textbf{Comput.} & \textbf{Other} \\ \midrule
\textbf{CPRP2P (baseline)}               & 0.95             & 63.12\%          & 28.43\%         & 8.38\%           & 0.06\%         \\
\textbf{C-Coll (baseline)}               & 1.39             & 53.47\%          & 34.24\%         & 12.17\%          & 0.11\%         \\ \midrule
\textbf{{\pname} (single-thread)} & 1.61             & 58.23\%          & 27.57\%         & 14.05\%          & 0.15\%         \\
\textbf{{\pname} (multi-thread)}  & 2.96             & 23.18\%          & 50.65\%         & 25.87\%          & 0.30\%         \\ \bottomrule
\end{tabular}%
}
\end{table}

Apart from the performance analysis, we also evaluate the numeric and visual accuracy of our {\pname}. With an error bound of 1E-4, {\pname} achieves a Peak Signal-to-Noise Ratio (PSNR)~\cite{PSNR} of 49.1 and a Normalized Root Mean Square Error (NRMSE)~\cite{shcherbakov2013errormetrics} of 3.5E-3, demonstrating excellent data quality. In Figure \ref{fig-image-stacking-quality}, we compare the visual accuracy of the {\pname} framework against the original MPI without compression. The comparison shows no visual difference between the two methods, further confirming the high accuracy of our {\pname} framework. In summary, {\pname} achieves sound speedups over the baselines while maintaining high data accuracy.

\begin{figure}[ht]
    \centering
    \subfloat[MPI (lossless)]{
        \includegraphics[scale=0.25]{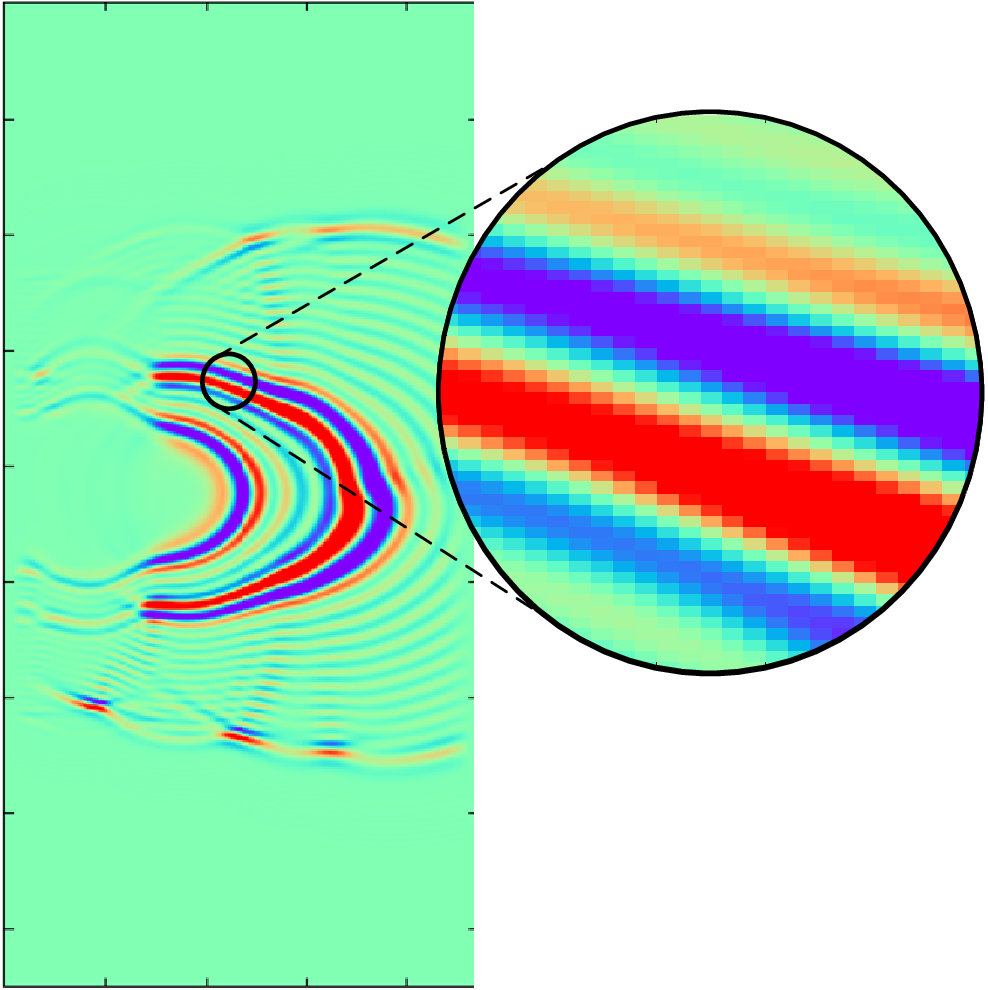}
        \label{fig-image-stacking-ori}
    }
    \hspace{-2mm}
    \subfloat[{\pname}]{
        \includegraphics[scale=0.25]{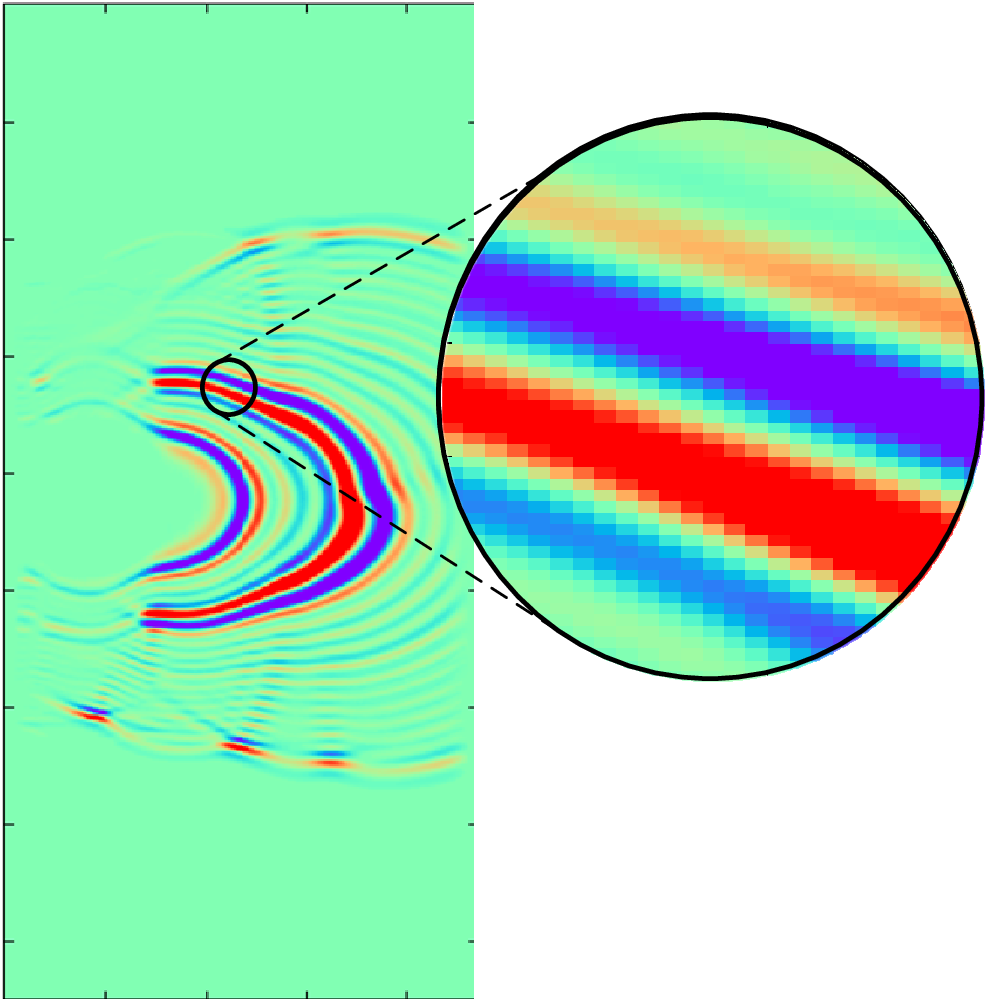}
        \label{fig-image-stacking-lossy}
    }
    \caption{Visualization of image stacking application.}
    \label{fig-image-stacking-quality}
\end{figure}


\section{Conclusion and Future Work}
\label{sec:conclusion}

In this paper, we introduce {\pname}, a novel design for lossy-compression-integrated MPI collectives that significantly improves performance with bounded errors. Our two proposed high-performance frameworks for compression-integrated MPI collectives, together with optimized and customized pipe-lined {\fzlight}, enable us to implement Z-Allreduce, which outperforms the original Allreduce by up to \textbf{3.6$\times$} while preserving high data quality. We demonstrate the generalizability of our approaches through Z-Scatter and Z-Bcast, which outperform the original MPI\_Scatter and MPI\_Bcast by up to \textbf{5.4$\times$} and \textbf{8.9$\times$}, respectively. In summary, our research has addressed the issues of sub-optimal performance, lack of generality, and unbounded errors in lossy-compression-integrated MPI collectives, laying the foundation for future research in this area. Moving forward, we plan to expand our research by implementing more {\pname} based collectives and deploying our frameworks and optimizations on other hardware, such as GPUs and AI accelerators.
\section{Acknowledgment}
This research was supported by the U.S. Department of Energy, Office of Science, Advanced Scientific Computing Research (ASCR), under contract DE-AC02-06CH11357, and supported by the National Science Foundation under Grant OAC-2003709, OAC-2104023, and OAC-2311875. The authors extend their appreciation to the Deanship of Scientific
Research at University of Bisha, Saudi Arabia for supporting this research work through the Promising
Program under Grant Number (UB- Promising -40 - 1445). The experimental resource for this paper was provided by the Laboratory Computing Resource Center on the Bebop cluster at Argonne National Laboratory.


\begin{IEEEbiography}
[{\includegraphics[width=1in,height=1.25in,clip,keepaspectratio]{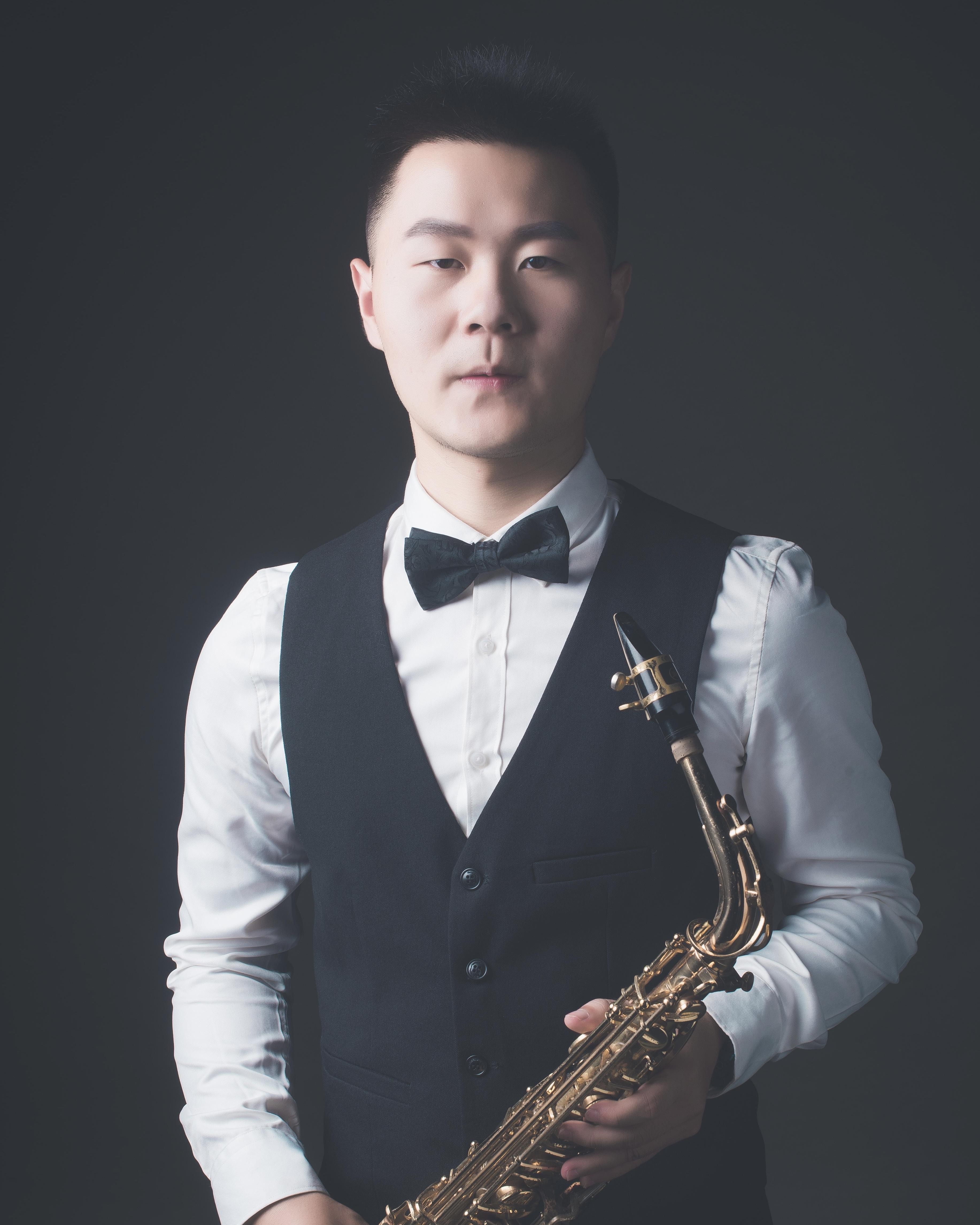}}]{Jiajun Huang} is a Ph.D. candidate in Computer Science at the University of California, Riverside, and a long-term visiting student at Argonne National Laboratory. He received his bachelor's degree in Electronic Information Engineering from the University of Electronic Science and Technology of China (UESTC) and the University of Glasgow (Honors of the First Class), in 2021. His research interests include distributed and parallel computing/systems, high-performance computing, and general artificial intelligence. Email: jhuan380@ucr.edu
\end{IEEEbiography}
\vskip -2\baselineskip plus -1fil

\begin{IEEEbiography}[{\includegraphics[width=1in,height=1.25in,clip,keepaspectratio]{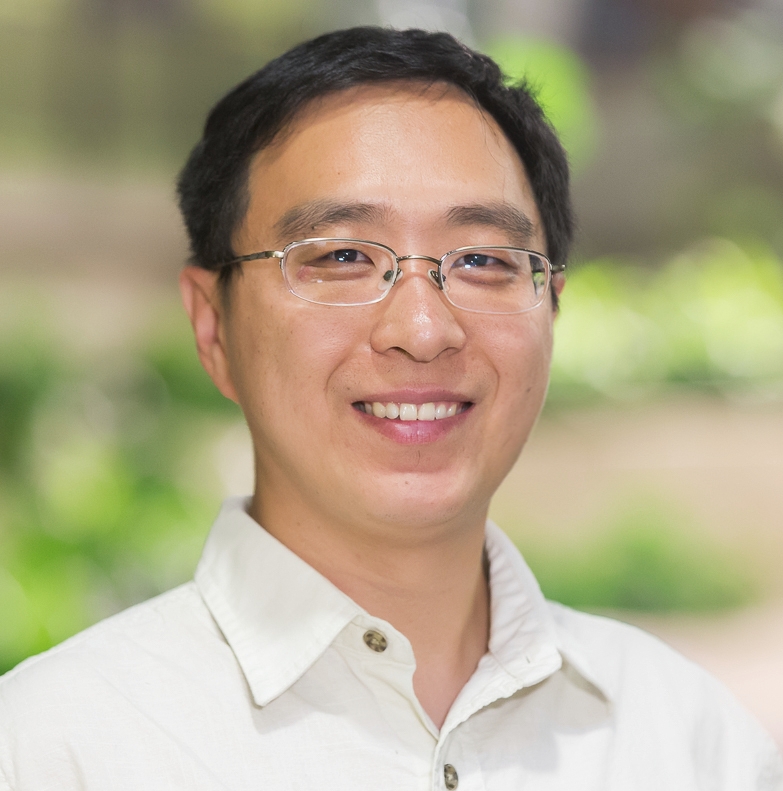}}]{Sheng Di}
(Senior Member, IEEE) received his master's degree from Huazhong University of Science and Technology in 2007 and Ph.D. degree from the University of Hong Kong in 2011. He is currently a computer scientist at Argonne National Laboratory. 
His research interests involve resilience on high-performance computing (such as silent data corruption, optimization checkpoint model, and in situ data compression) and broad research topics on cloud computing.
He is working on multiple HPC projects, such as detection of silent data corruption, characterization of failures and faults for HPC systems, and optimization of multilevel checkpoint models. He is the recipient of a DOE 2021 Early Career Research Program award. Email: sdi@anl.gov.
\end{IEEEbiography}

\vskip -2\baselineskip plus -1fil

\begin{IEEEbiography}[{\includegraphics[width=1in,height=1.25in,clip,keepaspectratio]{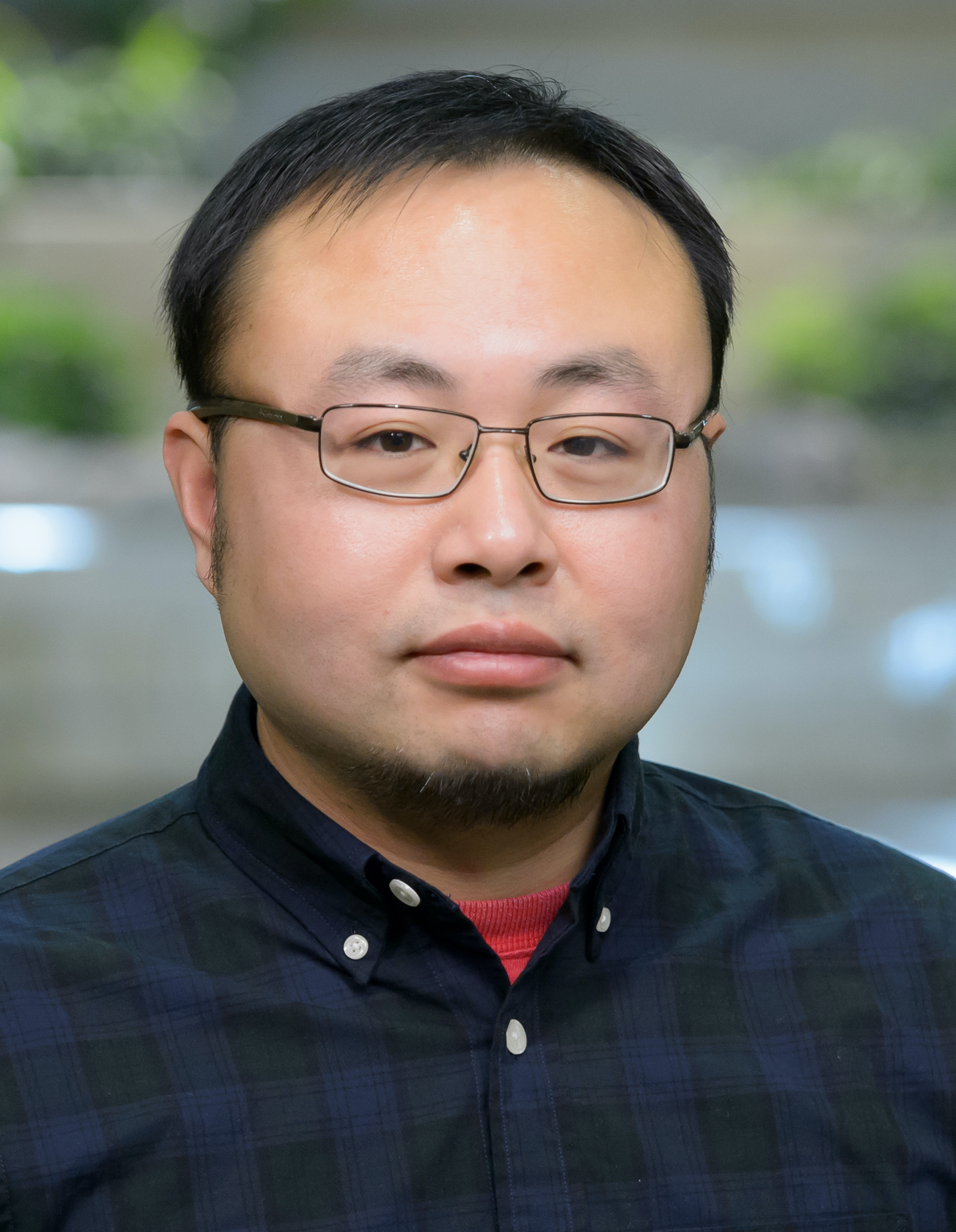}}]{Xiaodong Yu}
is an Assistant Professor in the Computer Science Department at Stevens Institute of Technology. He was an Assistant Computer Scientist at Argonne National Laboratory. He earned his Ph.D. in Computer Science from Virginia Tech in 2019. His research areas span parallel and distributed computing, next-generation AI hardware, and machine learning privacy and security. Email: xyu38@stevens.edu.
\end{IEEEbiography}

\vskip -2\baselineskip plus -1fil

\begin{IEEEbiography}[{\includegraphics[width=1in,height=1.25in,clip,keepaspectratio]{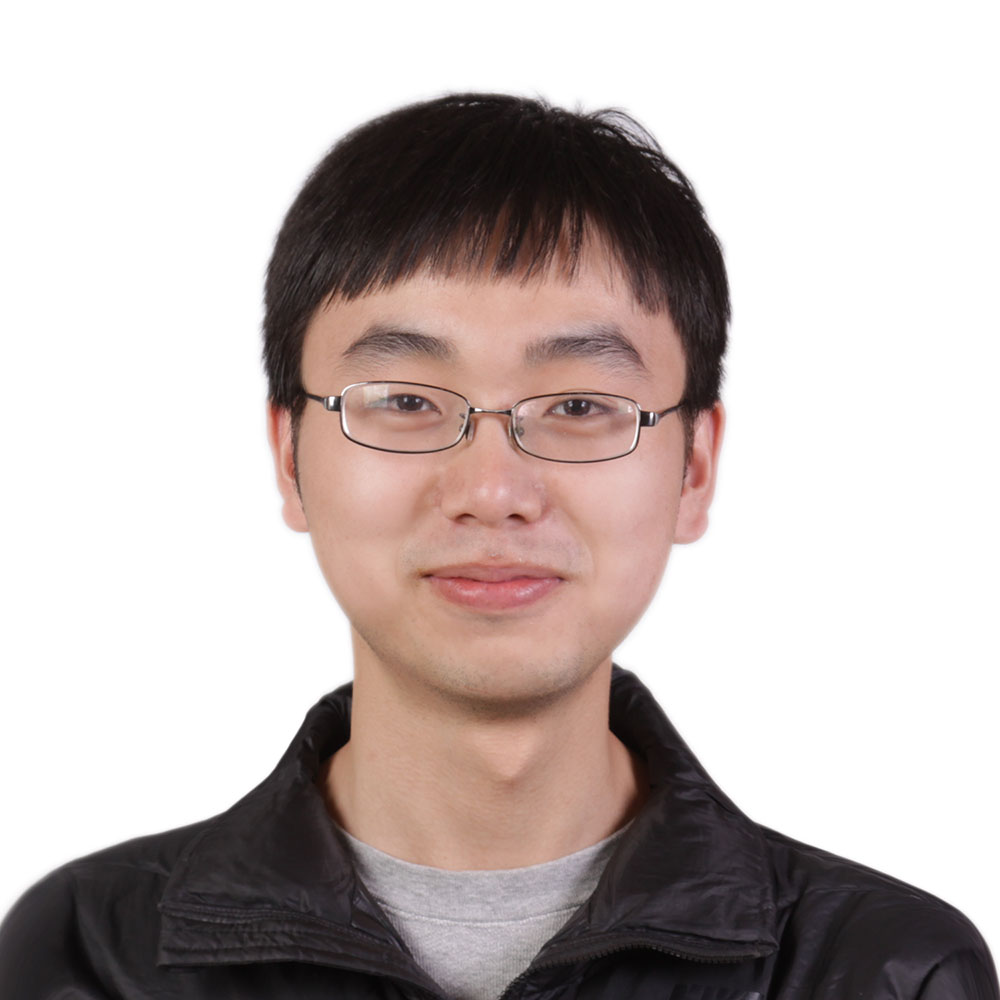}}]{Yujia Zhai} received his bachelor's degree from the University of Science and Technology of China in 2016, a master's degree from Duke University in 2018, and a Ph.D. degree from the University of California, Riverside in 2023. He is interested in performance optimization for math libraries on GPUs. Email: yzhai015@ucr.edu.
\end{IEEEbiography}

\vskip -2\baselineskip plus -1fil

\begin{IEEEbiography}[{\includegraphics[width=1in,height=1.25in,clip,keepaspectratio]{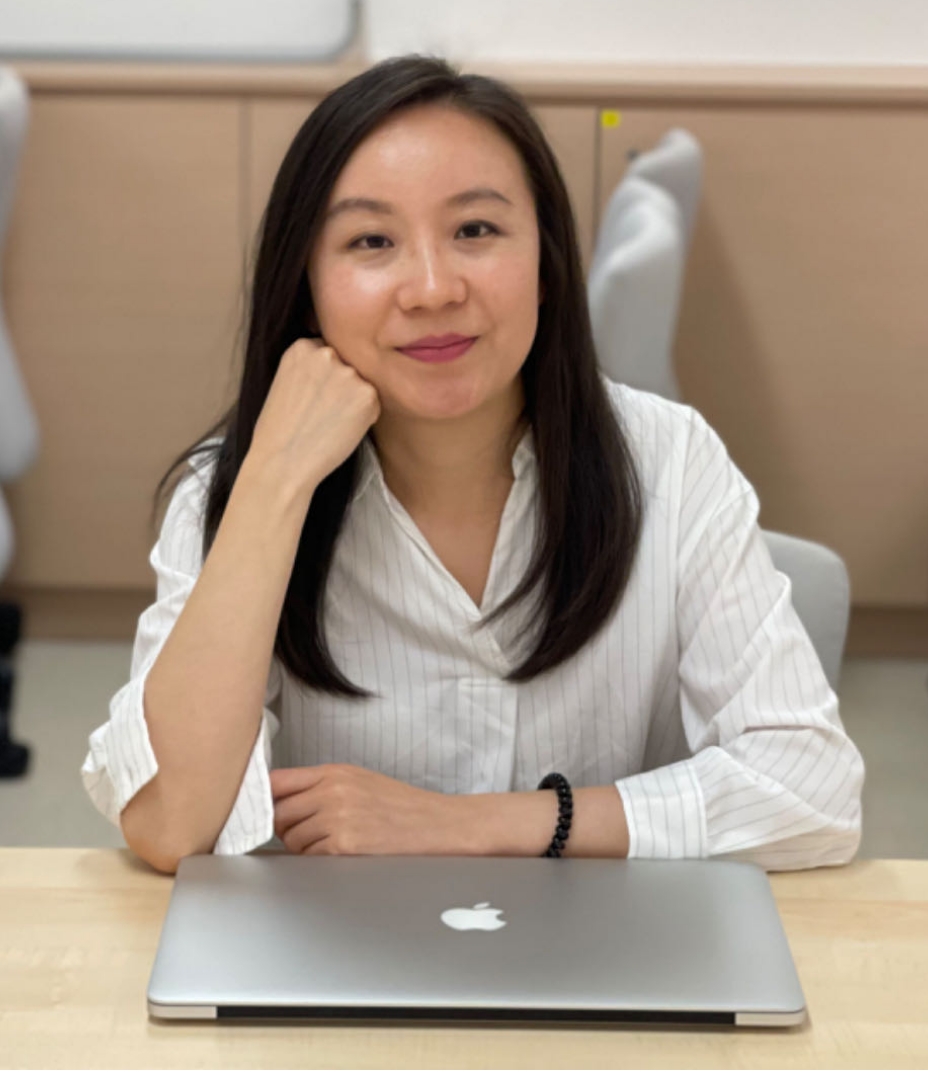}}]{Zhaorui Zhang} is currently a research assistant professor in the Department of Computing at The Hong Kong Polytechnic University. She received her Ph.D. from the Department of Computer Science at The University of Hong Kong, Hong Kong, and her BSc degree in computer science from Xi'an Jiaotong University. Her research interests include distributed machine learning systems, distributed systems, HPC, cloud computing, and data reduction. Email: zhaorui.zhang@polyu.edu.hk.
\end{IEEEbiography}

\vskip -2\baselineskip plus -1fil

\begin{IEEEbiography}[{\includegraphics[width=1in,height=1.25in,clip,keepaspectratio]{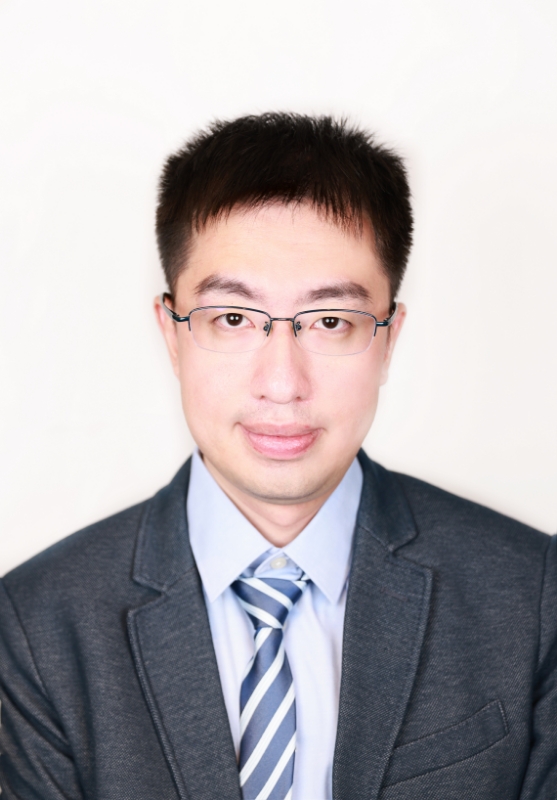}}]{Jinyang Liu} is an assistant professor at the department of Computer Science in the University of Houston. Jinyang’s research lies in the interdisciplinary areas of High-Performance Computing, Scientific Data Management, and Artificial Intelligence. He has multiple published or accepted works in various highly prestigious conferences and journals such as ACM SIGMOD, IEEE/ACM SC, ACM ICS (one paper in the best paper finalist), IEEE ICDE, IEEE Cluster, IEEE BigData, IEEE TPDS, etc.
\end{IEEEbiography}

\vskip -2\baselineskip plus -1fil

\begin{IEEEbiography}[{\includegraphics[width=1in,height=1.25in,clip,keepaspectratio]{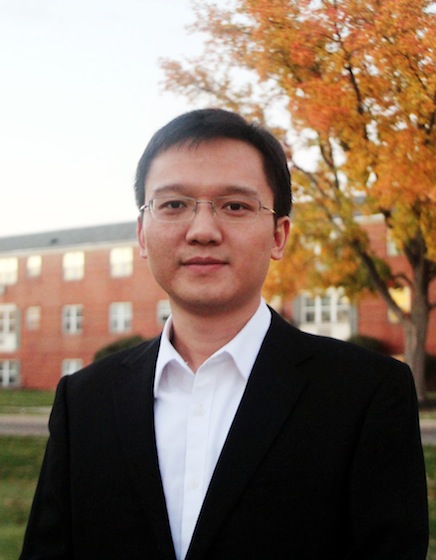}}]{Xiaoyi Lu} is an Associate Professor in the Department of Computer Science and Engineering at the University of California, Merced. His current research interests include parallel and distributed computing, high-performance networking and I/O technologies, big data analytics, cloud computing, and deep learning. He has published one book and more than 150 papers in prestigious international conferences, workshops, and journals with multiple Best (Student) Paper Awards or Nominations. Dr. Lu has received the NSF CAREER Award and other research awards from Meta, Amazon, and Google. Email:  xiaoyi.lu@ucmerced.edu.
\end{IEEEbiography}

\vskip -2\baselineskip plus -1fil

\begin{IEEEbiography}[{\includegraphics[width=1in,height=1.25in,clip,keepaspectratio]{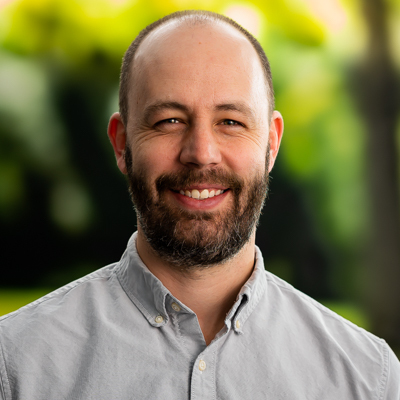}}]{Ken Raffenetti} is a Principal Software Development Specialist in the Programming Models and Runtime Systems group at Argonne National Laboratory. He is a core MPICH developer and active participant in several HPC industry working groups. Prior to Argonne, Ken earned a B.S. in Computer Science from the University of Illinois Urbana-Champaign.
\end{IEEEbiography}

\vskip -2\baselineskip plus -1fil

\begin{IEEEbiography}[{\includegraphics[width=1in,height=1.25in,clip,keepaspectratio]{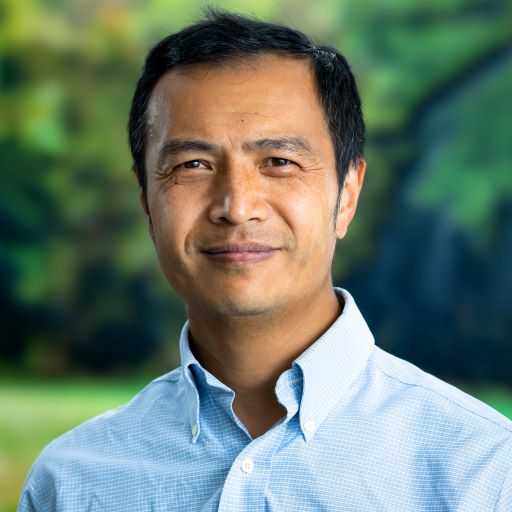}}]{Hui Zhou} is a Principal Research Software Engineer at Argonne National Laboratory and a core member of the MPICH development team. His research focuses on runtime systems for high-performance computing, accessible parallel computing, and scalable software development. He has a particular interest in enhancing interoperability between runtime systems.
\end{IEEEbiography}

\vskip -2\baselineskip plus -1fil

\begin{IEEEbiography}[{\includegraphics[width=1in,height=1.25in,clip,keepaspectratio]{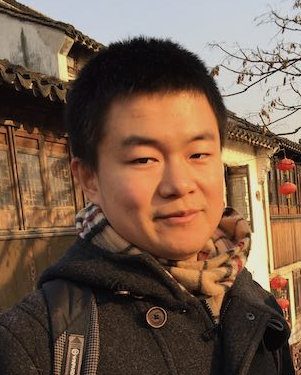}}]{Kai Zhao} is an assistant professor in the computer science department at Florida State University. He
received his bachelor's degree from Peking University in 2014 and his Ph.D. degree from the University of California, Riverside in 2022.
His research interests include high-performance computing and scientific data management. Email: kai.zhao@fsu.edu.
\end{IEEEbiography}

\vskip -2\baselineskip plus -1fil

\begin{IEEEbiography}[{\includegraphics[width=0.9in,height=1.25in,clip,keepaspectratio]{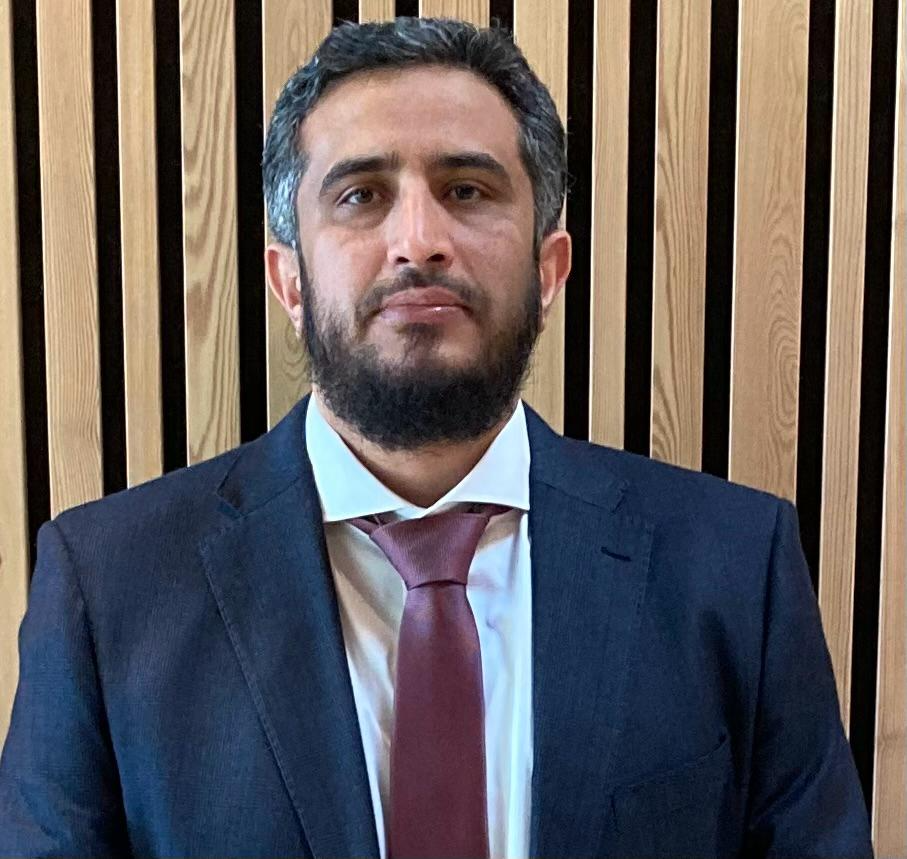}}]{Khalid Ayed Alharthi} is an assistant professor at the University of Bisha, KSA. He received his bachelor's degree from King Khalid University in 2008, his master's degree from Kent State University, USA in 2013, and his Ph.D. from the University of Warwick, UK in 2023. He is a long-term intern at UChicago Argonne National Laboratory and the Alan Turing Institute, UK. His research interests include AI, NLP, and AI in supporting resilience in HPC systems. Email: kharthi@ub.edu.sa.
\end{IEEEbiography}

\vskip -2\baselineskip plus -1fil

\begin{IEEEbiography}[{\includegraphics[width=1in,height=1.25in,clip,keepaspectratio]{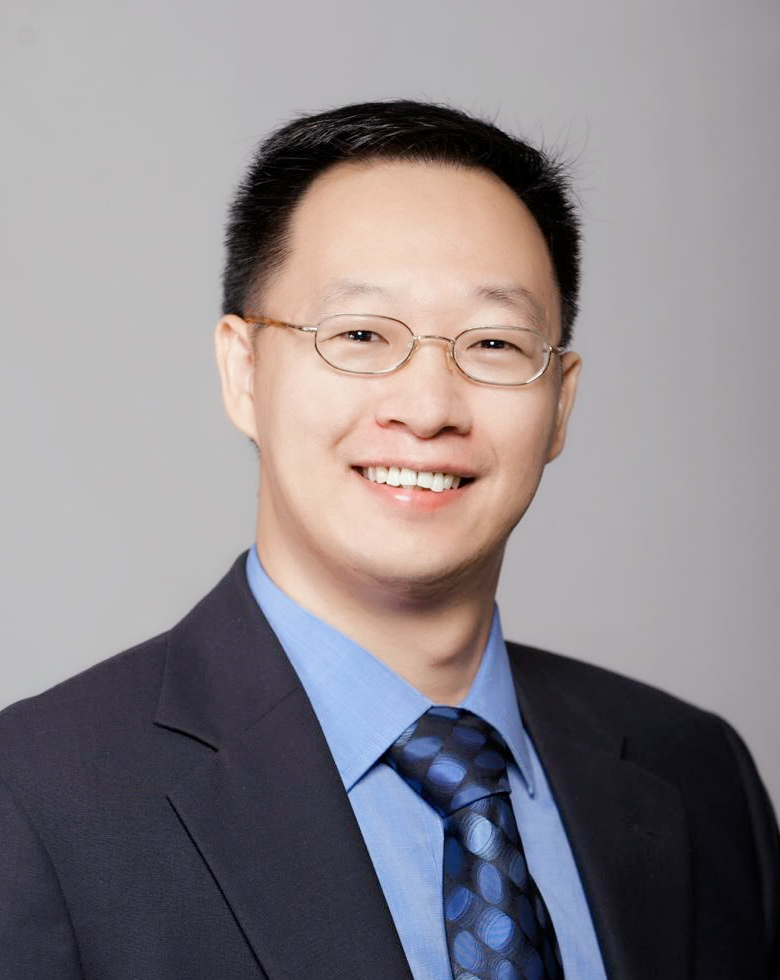}}]{Zizhong Chen}
(Senior Member, IEEE) received a bachelor's degree in mathematics from Beijing Normal University, a master's degree in economics from the Renmin University of China, and a Ph.D. degree in computer science from the University of Tennessee, Knoxville. He is a professor of computer science at the University of California, Riverside. 
His research interests include high-performance computing, parallel and distributed systems, big data analytics, cluster and cloud computing, algorithm-based fault tolerance, power and energy efficient computing, numerical algorithms and software, and large-scale computer simulations. His research has been supported by the National Science Foundation, Department of Energy, CMG Reservoir Simulation Foundation, Abu Dhabi National Oil Company, Nvidia, and Microsoft Corporation. 
He received a CAREER Award from the US National Science Foundation and a Best Paper Award from the International Supercomputing Conference. He is a Senior Member of the IEEE and a Life Member of the ACM. Email: chen@cs.ucr.edu.
\end{IEEEbiography}

\vskip -2\baselineskip plus -1fil

\begin{IEEEbiography}[{\includegraphics[width=1in,height=1.25in,clip,keepaspectratio]{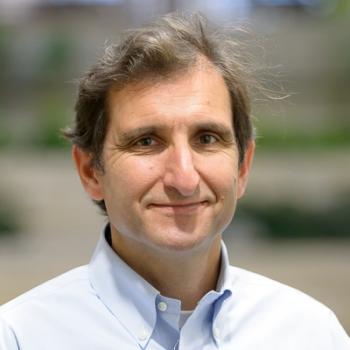}}]{Franck Cappello}
is the director of the Joint-Laboratory on Extreme Scale Computing gathering six of the leading high-performance computing institutions in the world: Argonne National Laboratory, National Center for Scientific Applications, Inria, Barcelona Supercomputing Center, Julich Supercomputing Center, and Riken AICS. He is a senior computer scientist at Argonne National Laboratory and an adjunct associate professor in the Department of Computer Science at the University of Illinois at Urbana-Champaign. He is an expert in resilience and fault tolerance for scientific computing and data analytics. Recently he started investigating lossy compression for scientific datasets to respond to the pressing needs of scientist performing large-scale simulations and experiments. His contribution to this domain is one of the best lossy compressors for scientific datasets respecting user-set error bounds. He is a member of the editorial board of the \textit{IEEE Transactions on Parallel and Distributed Computing} and of the \textit{ACM HPDC} and \textit{IEEE CCGRID} steering committees. He is a fellow of the IEEE. Email: cappello@mcs.anl.gov.
\end{IEEEbiography}

\vskip -2\baselineskip plus -1fil

\begin{IEEEbiography}[{\includegraphics[width=1in,height=1.25in,clip,keepaspectratio]{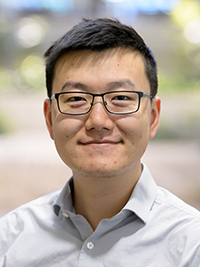}}] {Yanfei Guo} holds an appointment as an Computer Scientist
at the Argonne National Laboratory. He is a member of the Programming
Models and the Runtime Systems group. He has been working on multiple
software projects including MPI, Yaksa and OSHMPI. His research
interests include parallel programming models and runtime systems in
extreme-scale supercomputing systems, data-intensive computing and cloud
computing systems. Yanfei has received the best paper award at the
USENIX International Conference on Autonomic Computing 2013 (ICAC'13).
His work on programming models and runtime systems has been published on
peer-reviewed conferences and journals including the ACM/IEEE
Supercomputing Conference (SC'14, SC'15) and IEEE Transactions on
Parallel and Distributed Systems (TPDS). Yanfei have delivered eight
tutorials on MPI to various audience levels from university students to
researchers. Yanfei served as reviewers and technical committee members
in many journals and conferences. He is a member of the IEEE and a
member the ACM.
\end{IEEEbiography}

\vskip -2\baselineskip plus -1fil

\begin{IEEEbiography}[{\includegraphics[width=1in,height=1.25in,clip,keepaspectratio]{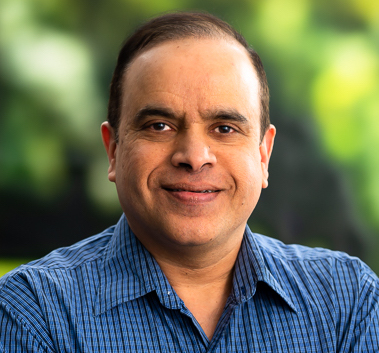}}]{Rajeev Thakur} is an Argonne Distinguished Fellow and Deputy Director of the Data Science and Learning Division at Argonne National Laboratory. He received a Ph.D. in Computer Engineering from Syracuse University. His research interests are in high-performance computing, parallel programming models, runtime systems, communication libraries, scalable parallel I/O, and artificial intelligence and machine learning. He is a Fellow of IEEE.
\end{IEEEbiography}

\end{document}